\documentclass[aps,pra,reprint]{revtex4-2}
\usepackage{graphicx}
\usepackage[dvipsnames]{xcolor}
\usepackage{hyperref}
\usepackage{physics}
\hypersetup{
colorlinks=true,
citecolor=magenta}
\usepackage{amsmath}
\usepackage{amssymb}
\usepackage{amsthm}
\theoremstyle{remark}
\newtheorem{theorem}{Theorem}
\newtheorem{proposition}[theorem]{Proposition}
\newtheorem{lemma}[theorem]{Lemma}

\newtheorem{remark}[theorem]{Remark}

\usepackage{tikz}
\usepackage{quantikz}

\makeatletter
\DeclareRobustCommand{\equalcontribution}{%
  \frontmatter@footnote{These authors contributed equally to this work.}}
\makeatother

\begin{document}
\title{Quantum Codeword Sensing}
\author{John Blue\equalcontribution}\affiliation{Department of Physics, Massachusetts Institute of Technology, Cambridge, Massachusetts 02139, USA}
\author{Zachary E. Chin\equalcontribution}\email{zchin@mit.edu}\affiliation{Department of Physics, Massachusetts Institute of Technology, Cambridge, Massachusetts 02139, USA}
\author{Vlad S. Oros}\affiliation{Department of Physics, Massachusetts Institute of Technology, Cambridge, Massachusetts 02139, USA}
\author{Isaac L. Chuang}\affiliation{Department of Physics, Massachusetts Institute of Technology, Cambridge, Massachusetts 02139, USA}
\date{\today}
\begin{abstract}
    Consider the task of identifying an unknown classical codeword delivered to a register of quantum sensor qubits through a sequence of weak phase interactions that each encode one bit. Suppose that each received bit controls an identical $Z$-rotation of angle $\theta$ applied to every sensor qubit; receivers may also use auxiliary qubits within a total qubit budget $Q$, with arbitrary control unitaries and entanglement among all qubits allowed during each reception step. When the total collected signal is weak ($Q\theta \rightarrow 0$), we establish that a sensor coupled to a persistent quantum memory can identify the original codeword more reliably than any receiver with the same $Q$ and $\theta$ but only classical memory between received bits, even when bit flips occur during transmission. In particular, for a family of $[n,k]$ binary simplex codes with even $k$ and $n=2^k-1$, we achieve exact correction of up to $t\leq(k-2)/4$ transmission errors by coupling the sensor register to a persistent quantum memory of the same size, using $Q = \Theta((t+1)\log n)$ and $\theta=\frac{2\pi}{k\sqrt{n+1}}$; for $t=o(\sqrt k)$, every comparable receiver with only classical memory instead has average success probability (over uniformly chosen codewords) tending to zero. Furthermore, under independent bit flips of any fixed probability $p<1/16$, our receiver with persistent quantum memory achieves decoding error tending to zero as $n\to\infty$ using $Q=O(\log n\log\log n)$ total qubits and the same $\theta$. In contrast, every otherwise equivalent receiver with only classical memory has vanishing average success probability under the same stochastic error model.
    
\end{abstract}
\maketitle

\section{Introduction}
\label{sec:intro}
Suppose a classical message is encoded into a binary codeword whose bits arrive sequentially at a quantum sensor. Each received zero results in no action, and each received one applies a \(Z\) rotation of angle \(\theta\) to every sensor qubit. The receiver's task is to determine the original message, even when some of the bits may be flipped during transmission. We call this task \emph{quantum codeword sensing}. 


Quantum codeword sensing models a situation that arises whenever a quantum sensor must read a classical signal that is structured in time but couples only weakly to the sensor. For example, a qubit dispersively coupled to a microwave cavity picks up a phase that depends on whether a photon is present. Single-photon detectors engineer this phase to equal $\pi$~\cite{kono2018itinerantphoton,dixit2021darkmatter}, but it is generally smaller when the coupling strength or interaction time is limited. Likewise, spin qubits such as nitrogen-vacancy centers or trapped ions pick up a phase proportional to an applied magnetic-field pulse~\cite{degen2017quantumsensing}, and Rydberg-atom receivers already decode amplitude-modulated microwave carriers as digital communication channels~\cite{meyer2018digitalcommunication,fancher2021rydbergsensors}. In each case an on--off keyed signal fits the model above exactly: a received one applies a collective $Z$ rotation to the sensor register, and a received zero does nothing. The angle $\theta$ is set by the physics of the coupling, such as the distance to the source, the transmitted power, or the dispersive shift. It may lie far below $\pi$, so that no single arrival can be resolved reliably. The message must then be spread redundantly across many arrivals, and a classical error-correcting code does this naturally while also protecting against bits that are flipped before they reach the sensor. A closely related spatial problem is quantum trajectory sensing~\cite{PhysRevLett.134.210802, chin2025trajectorysensing}, in which particles crossing a sensor array leave phases whose pattern must be identified in a single shot. Quantum codeword sensing can be viewed as its temporal counterpart: the pattern is spread over time slots instead of over sensor locations.

We may equip the receiver with different levels of quantum capabilities.
At the most basic level, it can prepare and measure its qubits individually, retaining only classical measurement outcomes between arrivals. A more capable receiver may prepare entangled states and perform joint measurements within each time slot, but still only carry classical information between slots. Finally, a receiver with a persistent quantum memory may preserve quantum information throughout the protocol and process signals coherently as they arrive. When \(\theta = \pi\), a single qubit can determine the received bit with certainty, and the problem reduces to classical decoding. In this work, we are therefore interested in understanding, at a fixed total qubit count \(Q\) and small total strength \(\Phi = Q\theta\), the extent to which these different capabilities improve decoding performance. In particular, we might expect that persistent quantum memory could offer an advantage over receivers that can use arbitrary entangled probes and joint measurements within each time slot, but only carry classical information between time slots.
%
It is natural to anticipate that increasing the quantum resources should allow for increased decoder success probability. For example, by preparing a GHZ state, entanglement can readily improve a receiver's ability to determine each incoming bit when \(\theta\) is small~\cite{degen2017quantumsensing}.


Beyond enhancing sensitivity to individual signals, quantum resources can also be used to \emph{process} signals coherently before they are measured. This idea, recently termed quantum computational sensing~\cite{khan2025qcsadvantage,khan2026qcsqsp,sen2026qcsmagnetic}, starts from the fact that measuring a sensor yields only limited classical information, at most one bit per qubit. When the goal is a decision about the signal rather than an estimate of it, it can therefore be worthwhile to compute on the sensor's quantum state before reading it out~\cite{sarovar2023computationalimaging}. Quantum signal processing interleaves signal interactions with control rotations to implement nonlinear functions of an unknown phase, which enables single-shot binary decisions~\cite{sinanan-singh2024singleshotqsp,khan2026qcsqsp}; this idea has already been used to discriminate quantum channels in a trapped-ion experiment~\cite{debry2023channeldiscrimination}. Similarly, Grover-type search driven directly by an unknown oscillating field can detect that field faster than protocols built from conventional sensing and classical post-processing~\cite{allen2025quantumcomputingenhanced}. In parallel, results in quantum learning theory show that experiments which store quantum data in a quantum memory and process them coherently can be exponentially more efficient than experiments that measure each sample as it arrives~\cite{aharonov2022quantumalgorithmicmeasurement,chen2022exponentialseparations,huang2022quantumadvantagelearning}. Experimentally, an auxiliary nuclear-spin memory has been shown to improve the sensitivity of a solid-state spin sensor~\cite{zaiser2016quantummemory}. Together, these results suggest that what matters is not only whether a sensor is entangled, but whether quantum information is preserved and processed between interactions. This raises the question: when a classical codeword is delivered through a stream of weak signal interactions, can a receiver with a persistent quantum memory outperform every receiver that carries only classical information from one arrival to the next, however much entanglement it uses within each time slot?

In this work, we answer this question in the affirmative. More specifically, we construct receivers for binary simplex codes~\cite{eczoo_simplex} which use coherent interactions between a quantum memory of size \(k = \log (n+1)\), and a sensor of the same size, to accumulate the small phases imparted by the rotations into orthogonal memory states. We then show that reliable decoding remains possible in the presence of errors. By deriving bounds on all receivers retaining only classical information between arrivals, we establish parameter regimes such that in the limit of large code size, our receiver's success probability approaches one while any classical memory receiver's approaches zero at the same qubit count and rotation angle. We further evaluate performance for finite \(k\) ranging from 6 to 16, using a rotation angle of \(\theta = 2\pi/((Q-k) \cdot 2^{k/2})\), finding that a receiver with \(Q = 40\) qubits can outperform a concrete non-adaptive classical memory strategy on every tested simplex code at small transmission error rates \(p < 0.01\).  And with \(Q = 80\) qubits, our receiver can outperform any classical memory receiver for even \(k\) between \(12\) and \(16\), again with transmission error rates below \(1\%\).

Two particular pioneering lines of work foreshadowed the scenario of quantum codeword sensing.
Perhaps most directly relevant is the task of quantum reading \cite{pirandola2011quantumreadingclassical,pirandola2011quantumreadingcapacity,das2019quantumreading}, in which a classical codeword is encoded into a sequence of channels. A receiver prepares a quantum state and then may apply each of the channels once to it, before making a measurement to determine the codeword. Indeed, the protocols we describe in this work may be considered examples of quantum reading, and it is known that quantum reading is assisted by the use of entangled states~\cite{pirandola2011quantumreadingclassical} and collective measurements~\cite{pirandola2011quantumreadingcapacity}. Quantum reading was initially motivated by the scenario of reading data from a classical digital memory. With instead a communication setting in mind, we specialize this framework to a stream of weak signal interactions that repeatedly act on the same sensor register, with a particular interest in receivers whose total qubit budget is too small to retain a separate quantum system for every received symbol.

This sequential mode of access has a direct antecedent in quantum data gathering~\cite{blume-kohout2013quantumdata}, which considers, among other tasks, discrimination between $2^k$ possible pure product states. Optimal discrimination can require a collective measurement, yet the authors show that such a measurement can be implemented by sequentially interacting each subsystem with a single $2^k$-dimensional quantum memory, and then performing a final measurement on that memory. This demonstrates that the information needed for discrimination can be gathered coherently without retaining every subsystem. Our construction for the simplex code realizes this principle by accumulating the effects of weak signal interactions into orthogonal memory states for distinct codewords.

The structure of the paper is as follows. We begin by formally describing the quantum codeword sensing problem and proving bounds on the success probability of receivers with only classical memory (Section~\ref{sec:describe-problem}). We then describe our protocol for the simplex code for a receiver with quantum memory (Section~\ref{sec:describe-protocol}), prove its advantage over any classical memory receiver and investigate its performance at small \(k\) (Section~\ref{sec:theorems}), 
before concluding (Section~\ref{sec:conclusion}). In Appendix~\ref{sec:deferred-proofs} we collect deferred proofs.

\section{The Quantum Codeword Sensing Problem}
\label{sec:describe-problem}
In this section we formally define the quantum codeword sensing problem and define three classes of receivers. We then derive bounds on the maximum success probabilities of receivers equipped only with classical memories. 

\begin{figure*}[ht]
  \centering
  \includegraphics[width=1.0\linewidth]{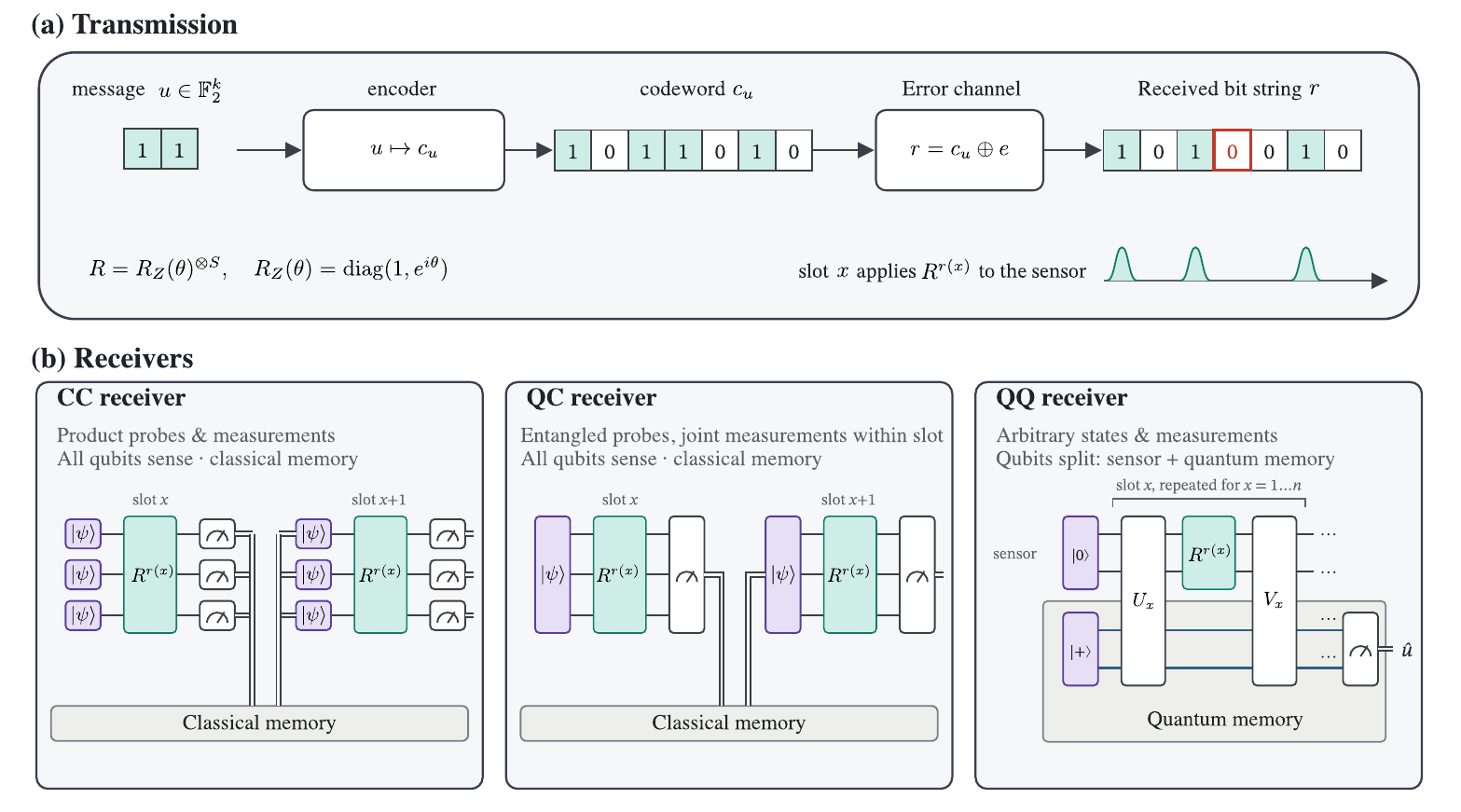}
  \caption{Quantum codeword sensing and receiver classes. (a) A message \(u\) is encoded as a classical codeword \(c_u\), which is corrupted by transmission errors \(e\). The received bits \(r=c_u\oplus e\) arrive sequentially: in time slot \(x\), the bit \(r(x)\) applies \(R^{r(x)}\) to the sensor, where \(R=R_Z(\theta)^{\otimes S}\) acts on its \(S \leq Q\) qubits. (b) Receivers are compared at the same total qubit budget \(Q\) and signal strength \(\theta\). CC receivers use product probes and product measurements; QC receivers may use entangled probes and joint measurements within each slot. Both retain only classical information between slots and prepare fresh probes, with preparations and measurements allowed to depend on previous outcomes. QQ receivers allocate qubits to sensing and persistent quantum memory. The illustrated QQ circuit uses controls \(U_x\) and \(V_x\) to couple the sensor and memory before and after each signal interaction, preserving quantum information across arrivals for a final measurement and estimate \(\hat u\) of the original message.}
  \label{fig:setup}
\end{figure*}

The setup of quantum codeword sensing is depicted in Fig.~\ref{fig:setup}. Let \(C \subset \mathbb{F}_2^n\) be an \([n, k, d]\) linear code over \(\mathbb{F}_2\), that is, \(C\) is a \(k\)-dimensional subspace of \(\mathbb{F}_2^n\), and the minimum Hamming weight of any nonzero element of \(C\) is \(d\). For each message \(u \in \mathbb{F}_2^k\), there is an associated codeword \(c_u \in C\). We assume a sender picks messages \(u\) uniformly at random, and sends the bits of the corresponding codeword over a potentially noisy channel one-by-one. For \(x \in [n]\), let \(c_u(x)\) denote the \(x^{\rm th}\) bit of \(c_u\), and let \(e(x)\), for \(e \in \{0, 1\}^n\), indicate the presence of an error during transmission. Let \(r = c_u \oplus e\) denote the bitstring after the transmission channel. At time slot \(x\), the sensor register with \(S\) qubits is acted on by the unitary \(R^{r(x)}\), where we call \(\theta\) the \emph{signal strength}, and \(R = R_Z(\theta)^{\otimes S}\), with
\begin{equation}
    R_Z(\theta) = \begin{pmatrix}
        1 & 0 \\
        0 & e^{i\theta}
    \end{pmatrix}.
\end{equation}
Given a computational basis state \(\ket{z}\), \(R\ket{z} = e^{i \theta |z|}\ket{z}\), where \(|z|\) denotes the Hamming weight of the bitstring \(z\).

A receiver is equipped with \(Q\) total qubits which  may be allocated towards sensing and memory. During the protocol, the receiver performs some sequence of unitary operations and measurements, producing some transcript \(y\). The  receiver then runs an algorithm which takes in \(y\), and outputs a predicted message \(\hat{u}\). The receiver succeeds if \(\hat{u} = u\).

We consider three possible combinations of sensor and memory capabilities. 
\begin{enumerate}
    \item A CC receiver (standing for ``classical-classical") may only maintain a classical memory between time slots, and thus dedicates all \(Q\) qubits towards sensing, and must re-initialize the state of its qubits at the beginning of each time slot. Additionally, a CC receiver may only prepare product states, and make product measurements. We allow the state preparations and measurements to be adaptive.
    \item A QC receiver (standing for ``quantum-classical") may also only maintain a classical memory between time slots (and thus must re-initialize its qubits at the start of each slot), and so we assume it dedicates all of its qubits towards sensing, but may prepare arbitrary states and make arbitrary measurements. Again, these states and measurements may be adaptive.
    \item Finally, a QQ receiver (standing for ``quantum-quantum") may prepare arbitrary quantum states, make arbitrary measurements, and may also maintain a persistent quantum memory (not exposed to the signal) throughout the entirety of the protocol which can interact coherently with the sensor. A QQ receiver must decide how to allocate its \(Q\) qubits between sensing and memory.
\end{enumerate}
We call \(\Phi = Q\theta\) the \emph{total strength} of a protocol. For the rest of this paper, unless otherwise stated, we assume \(\Phi \in [0, \pi]\).

Before continuing, we briefly collect notation that will be used throughout the paper. We use \(h_2(x) = -x \log(x) - (1-x)\log (1-x)\) to denote the binary entropy, where \(\log\) indicates the base two logarithm. For a positive integer \(n\), we use \([n]\) to denote the set \(\{1, \dots, n\}\). For an integer \(x \in \{0, \dots, n\}\), we often abuse notation and use the same symbol to refer to the integer, and its binary representation on \(\lceil \log (n+1) \rceil\) bits. Table \ref{tab:common-symbols} lists some commonly used symbols with their definitions.
\begin{table}[t]
\caption{Commonly used symbols}
\label{tab:common-symbols}
\centering
\begin{ruledtabular}
\begin{tabular}{lp{0.68\columnwidth}}
Symbol & Definition \\
\hline
\(Q\) & Total number of receiver qubits \\
\(S\) & Number of qubits per sensor register \\
\(B\) & Number of memory--sensor register pairs \\
\(\theta\) & Signal strength \\
\(S\theta\) & Interaction strength \\
\(\Phi=Q\theta\) & Total strength \\
\(n\) & Codeword length \\
\(k\) & Message length \\
\(u\) & Message \\
\(c_u\) & Codeword corresponding to message \(u\) \\
\end{tabular}
\end{ruledtabular}
\end{table}

\subsection{General Bounds}
\label{sec:cc-qc-bounds}
We now derive upper bounds on the success probabilities of CC and QC receivers that hold for any code. First, a small lemma.
\begin{lemma}
  \label{lem:overlap}
  For any \(m\) qubit state \(\ket{\psi}\) and angle \(\theta\), with \(m\theta \in[0, \pi]\),
  \begin{equation}
    |\langle \psi | R_Z(\theta)^{\otimes m} | \psi \rangle | \ge \cos \left ( \frac{m\theta}{2} \right ).
  \end{equation}
\end{lemma}

\begin{proof}
    Let \(\ket{\psi} = \sum_{z} \alpha_z \ket{z}\), so that \(\langle \psi|R_Z(\theta)^{\otimes m}|\psi\rangle = \sum_z |\alpha_z|^2 e^{i\theta |z|}\). This is a convex combination of points lying on the upper half of the unit circle in the complex plane, and thus is minimized for an equal sum of the two extreme phases \(0\) and \(m\theta\), and so
    \begin{align}
     |\langle \psi | R_Z(\theta)^{\otimes m} | \psi \rangle | &\ge \frac{1}{2} |e^{im\theta} + 1| \\
     &= \cos \left (\frac{m\theta}{2} \right )
    \end{align}
    as desired.
\end{proof}

Now, we show that under very mild constraints on \(\theta\) and \(\Phi\), a perfect CC or QC receiver is impossible.
\begin{proposition}
  \label{prop:qc-perfect-recovery}
    For any code with \(k \ge 1\), perfect recovery with a QC receiver is impossible if \(\Phi < \pi\), even in the absence of transmission errors. For a CC receiver, perfect recovery is impossible if \(\theta < \pi\).
\end{proposition}

\begin{proof}
  Suppose a codeword is transmitted with no errors. Consider a quantum codeword sensing protocol with a QC or CC receiver which does achieve perfect recovery. Let \(h\) be the classical information available prior to the last time slot, and suppose there are two distinct messages \(u\) and \(u'\) such that \(\Pr[u | h] \ne 0\) and \(\Pr[u' | h] \ne 0\). Thus, the final measurement must distinguish them with certainty. Since \(h\) is the same in both scenarios, the initial state prepared by the receiver is identical, call it \(\ket{\psi_0}\) (a protocol involving a mixed state preparation can be implemented by randomly sampling a pure state and retaining its classical label as part of \(h\); ignoring this label reproduces the original protocol). If the final bits of \(c_u\) and \(c_{u'}\) agree, then the states are identical, and perfect discrimination is not possible. If the bits differ, then one of the messages results in the state \(R\ket{\psi_0}\), and in the QC case, Lemma~\ref{lem:overlap} gives that
\[|\langle \psi_0 | R |\psi_0\rangle | \ge \cos(\Phi/2)\]
and perfect discrimination is still not possible if \(\Phi < \pi\). Thus, the message must have already been determined prior to the final time slot. Repeating this argument implies that the message was determined prior to the first time slot, which implies that the code only contains a single codeword, a contradiction.

The argument is almost identical in the CC case, except now, \(\ket{\psi_0}\) is a product state, and Lemma~\ref{lem:overlap} gives that
\[|\langle \psi_0 | R |\psi_0\rangle | \ge \cos^Q(\theta/2)\]
and the states are not perfectly distinguishable when \(\theta < \pi\).
\end{proof}

As we will see in the following section, there are codes for which a QQ receiver may succeed with probability one, even when \(\Phi \ll \pi\). While this already demonstrates an advantage for QQ receivers, it's not necessarily a terribly interesting one, as a CC or QC receiver might still succeed with probability close enough to one to be indistinguishable in practice. In the following proposition, we find that there are regimes where this is not the case, by bounding the success probability of any QC or CC receiver, assuming the transmitted codeword is subject to a binary symmetric noise channel. 

\begin{proposition}
  \label{prop:qc-stochastic}
  Consider a quantum codeword sensing protocol with signal strength \(\theta\), for an \([n,k,d]\) code, using a QC receiver with \(Q\) qubits, with \(\Phi = Q\theta \in [0, \pi]\). Assuming that each codeword is transmitted with equal likelihood, and that during transmission, the bits of the codeword are flipped independently with probability \(p \in [0, 1/2]\), we have that the probability of successful decoding is upper bounded by
  \begin{equation}
    \label{eq:qc-upper-bound}
    \begin{split}
      \Pr[u = \hat{u}] \le \min \Bigg \{ &\frac{1}{2^k} + \sqrt{\frac{\ln 2}{2} n C_{QC}(\Phi, p)} , \\
      &\frac{1 + n C_{QC}(\Phi, p)}{k} \Bigg \}
    \end{split}
  \end{equation}
  where we assume \(k \ge 1\) and, for \(x \in [0,1]\), have written
  \begin{equation}
    \label{eq:C-generic}
    C(x, p) = 1-h_2\!\left(\frac{1-(1-2p)x}2 \right),
  \end{equation}
  so that \(C_{QC}(\Phi, p) = C\big(\sin(\Phi/2),\, p\big)\). For a CC receiver with the same \(Q\) and \(\theta\), the success probability is upper bounded by
  \begin{equation}
    \label{eq:cc-upper-bound}
    \begin{split}
      \Pr[u = \hat{u}] \le \min \Bigg \{ &\frac{1}{2^k} + \sqrt{\frac{\ln 2}{2} n C_{CC}(\theta, p, Q)} , \\
      &\frac{1 + nC_{CC}(\theta, p, Q)}{k} \Bigg \}
    \end{split}
  \end{equation}
  where now \(C_{CC}(\theta, p, Q) = C\big(v_Q(\theta),\, p\big)\), with
  \begin{equation}
    \label{eq:cc-visibility}
    v_Q(\theta) = \sqrt{1-\cos^{2Q}(\theta/2)}.
  \end{equation}
\end{proposition}

The rough idea of the proof is that in the case of a QC receiver, we may treat every time slot as a single use of the CQ channel
\begin{align}
    0 &\mapsto (1-p) |\psi_0\rangle\langle\psi_0| + p |\psi_1\rangle\langle\psi_1|\\
    1 &\mapsto (1-p) |\psi_1\rangle\langle\psi_1| + p |\psi_0\rangle\langle\psi_0|
\end{align}
where \(|\psi_0\rangle\) is the initial state of the sensor, and \(|\psi_1\rangle = R|\psi_0\rangle\). By upper bounding the accessible information per time slot and using Fano's inequality, we may upper bound the probability of successful decoding. We defer the details to Appendix~\ref{sec:proof-qc-stochastic}.

A few brief remarks to build intuition for Equations~\eqref{eq:qc-upper-bound} and \eqref{eq:cc-upper-bound}. First, we may think of \(C_{QC}(\Phi, p)\) as an upper bound on the number of bits of information that may be extracted per time slot by a QC receiver. For a fixed code, as either \(p \to 1/2\) or \(\Phi \to 0\), this quantity approaches zero, indicating that we must increase the length of the code to compensate.
In this work we are particularly interested in receivers that work with a small number of qubits, and a small signal strength.  Using that \(h_2(x) = 1 - (2\ln 2)^{-1} (1-2x)^2 + O((x-1/2)^4)\), we have that for small \(\Phi\),
\begin{equation}
      C_{QC}(\Phi, p) \approx \frac{1}{2\ln 2}(1-2p)^2 \frac{\Phi^2}{4}.
\end{equation}
Plugging this into Equation~\eqref{eq:qc-upper-bound}, we get that
\begin{equation}a
      \label{eq:small-angle-qc-approx}
      \Pr[u=\hat{u}] \lessapprox \frac{1}{k} + \frac{(1-2p)^2}{8 \ln 2} \frac{n\Phi^2}{k}.
\end{equation}
Thus, if \(n\Phi^2 \ll k\), the second branch is \(o(1)\) and the decoding success probability must vanish; in the more demanding regime \(n\Phi^2 = O(1)\) it is at most \(\approx 1/k\). The first branch instead reads \(\Pr[u=\hat u] \lessapprox \frac{1}{2^k} + \frac{(1-2p)\sqrt{n}\,\Phi}{4}\), so it becomes informative only once \(\sqrt{n}\,\Phi \ll 1\), where it pins the success probability near the \(1/2^k\) of random guessing.
As we will see for the simplex code in the following section, in this small \(\Phi\) regime, a receiver with a quantum memory may surpass the bounds established in Proposition~\ref{prop:qc-stochastic} . In the case that the memory is at least the size of the code, this is perhaps not surprising, and may be seen as a manifestation of the power of joint measurements over product measurements \cite{hausladen1996classicalinformation,guha2011structuredoptical}, as the receiver may effectively store each signal in its memory and then perform a large collective measurement at the end. However, even as \(\Phi \to 0\), our QQ receivers will have a quantum memory scaling as \(o(n)\), suggesting that they are somehow extracting the message from its encoding as it is sequentially given the bits, even as each pulse encodes less and less information.

\section{Quantum Codeword Sensing with a Quantum Memory}
\label{sec:describe-protocol}
In this section we construct QQ decoders for the simplex code. We first introduce the code, then use a small example to show how a quantum memory can convert a sequence of signal pulses into a readable encoding of the message. We next generalize this circuit to decode larger codes exactly at weaker total strengths using a number of qubits logarithmic in the codeword length. Finally, we show how parallel decoders can correct a bounded number of transmission bit flips.

\subsection{The simplex code}
We design our QQ decoder for the binary simplex code since it offers particularly useful algebraic structure. The $[n,k]$ simplex code is defined as follows. Choose $n=2^k-1$ so that the $n$ bits of each codeword can be indexed by the nonzero bitstrings \(x\in\mathbb{F}_2^k\). Then, given a length-$k$ message \(u\in\mathbb{F}_2^k\) to be encoded, the corresponding codeword bit at position $x$ is defined to be \(c_u(x)=u\cdot x\), where the dot product is evaluated modulo two. The resulting code is linear and has distance $2^{k-1}$. For example, when \(k=2\), the $n=3$ positions of each codeword can be indexed as \(x=01,10,11\), and we obtain the codewords
\begin{equation}\label{eq:simplex_words}
\begin{aligned}
c_{00}&=000, & c_{01}&=101,\\
c_{10}&=011, & c_{11}&=110
\end{aligned}
\end{equation}
for each possible message. 

\subsection{A toy example}
To illustrate how the message of an incoming codeword can be extracted using a quantum memory, we construct a simple QQ decoder for the \([3,2]\) simplex code. The decoder is given by the circuit in Fig.~\ref{fig:toy_example}, which uses a memory register of \(k=2\) qubits  initialized to \(\ket{+}^{\otimes 2}\), and a sensor register containing one qubit initialized to \(\ket{0}\). For now, assume that there are no transmission errors and set \(\theta=\pi\), so that each received codeword bit 1 applies \(R_Z(\pi)=Z\) to the sensor qubit. Suppose the original message $u$ is unknown, so that the received string is therefore any one of the possible codewords \(000,101,011,110\) given in Eq. (\ref{eq:simplex_words}). The goal is to perform some control operations in between each received bit so that the message $u$ is ultimately printed to the memory as a bitstring \(\ket{u}_X=H^{\otimes 2}\ket{u}\) in the $X$-basis. The following simple choice of control gates turns out to be sufficient. Let \(e_i\) denote the string with a single 1 in position \(i\), so that \(e_1=10\) and \(e_2=01\). Then, at time slot \(x=e_i\), a CNOT from memory qubit \(i\) to the sensor is applied immediately before and after the signal interaction.
\begin{figure}[t]
\centering
\input{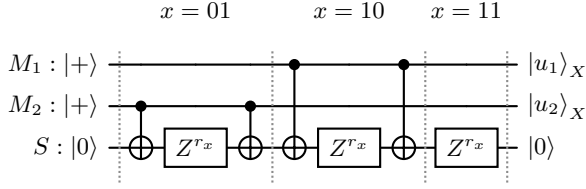}%
\caption{A QQ decoder for the \([3,2]\) simplex code at \(\theta=\pi\). Here \(r(x)=c_u(x)\) is the received bit at time step \(x\). In slot \(x=e_i\), two CNOTs couple memory qubit \(i\) to the sensor before and after the signal. The final memory state is \(\ket{u}_X\); a final \(X\)-basis measurement reads out \(u\).}
\label{fig:toy_example}
\end{figure}

To see why this choice of controls works, note that the $e_i$th bit of the received codeword always equals the $i$th bit of the underlying message, that is, \(c_u(e_i)=u_i\). Hence, if we could phase-flip the $i$th qubit in the memory whenever a 1 is received in the $e_i$th time slot, we could print the message $\ket{u}_X$ to the memory register, as desired. The CNOT controls used in this circuit implement this behavior exactly. For example, suppose that the message being sent is \(u=01\). The corresponding codeword being received is $c_u = 101$, so the signal applies \(Z\) in slots \(01\) and \(11\). Observe that a \(Z\) on the target of a CNOT can be commuted through the gate by applying an additional \(Z\) to its control, as shown in Fig.~\ref{fig:circuit_identities}(a). Figure~\ref{fig:pushing_Z} applies this identity to the toy circuit: moving the first signal \(Z\) through the preceding CNOT copies a \(Z\) onto memory qubit 2. The CNOTs then cancel in pairs, leaving two \(Z\) gates on the sensor qubit, which also cancel. Only the \(Z\) on memory qubit 2 remains, and this gate transforms the initial state $\ket{+}^{\otimes 2}$ into the desired message state $\ket{01}_X$. The same argument holds for any message \(u\): the total effective unitary applied to the memory is \(Z[u]=Z^{u_1}\otimes Z^{u_2}\), giving the final state
\begin{equation}
Z[u]\ket{+}^{\otimes 2}=\ket{u}_X.
\end{equation}
Thus, the circuit converts the temporal pattern of codeword bits into a spatial pattern of \(Z\) operations which represents the underlying message. The memory retains the phases acquired in successive time slots, allowing the message to be recovered in one final \(X\)-basis readout. Different messages produce orthogonal memory states and can therefore be distinguished with zero error.

\begin{figure}[t]
\centering
\input{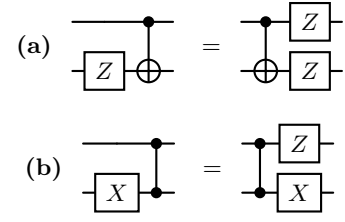}%
\caption{Pauli propagation identities. (a) Commuting a target \(Z\) in either direction through a CNOT adds a \(Z\) to its control. (b) Commuting an \(X\) in either direction through a controlled-\(Z\) adds a \(Z\) to the other qubit.}
\label{fig:circuit_identities}
\end{figure}

\begin{figure*}[t]
\centering
\small
\usetikzlibrary{arrows.meta}
\begin{tabular}{@{}c@{\hspace{6pt}}c@{\hspace{6pt}}c@{\hspace{6pt}}c@{\hspace{6pt}}c@{}}
{\begin{quantikz}[
 row sep={0.68cm,between origins},column sep=0.18cm,align equals at=2,
 execute at end picture={
  \coordinate (sensorpush) at ([xshift=-2.5pt,yshift=-6.5pt]group\tikzcdmatrixname-3-3.west);
  \draw[blue!65!black,line width=0.6pt,-{Stealth[length=3.5pt,width=3pt]}]
   (sensorpush) -- ([xshift=-8.5pt]\tikzcdmatrixname-3-2.center |- sensorpush);
  \draw[blue!65!black,line width=0.6pt,rounded corners=1.5pt,-{Stealth[length=3.5pt,width=3pt]}]
   ([xshift=-3pt,yshift=3.6pt]group\tikzcdmatrixname-3-3.west)
   -- ([xshift=4.3pt,yshift=3.6pt]\tikzcdmatrixname-3-2.center)
   -- ([xshift=4.3pt,yshift=3.4pt]\tikzcdmatrixname-2-2.center)
   -- ([xshift=-8.5pt,yshift=3.4pt]\tikzcdmatrixname-2-2.center);
 }]
\lstick{$M_1:\ket{+}$}
 & &[5pt] & & \ctrl{2} & & \ctrl{2} & & \\
\lstick{$M_2:\ket{+}$}
 & \ctrl{1} & & \ctrl{1} & & & & & \\
\lstick{$S:\ket{0}$}
 & \targ{} & \gate{Z} & \targ{} & \targ{}
 & & \targ{} & \gate{Z} &
\end{quantikz}}
 & \(=\) &
{\begin{quantikz}[row sep={0.68cm,between origins},column sep=0.22cm,align equals at=2]
\lstick{$\ket{+}$}
 & & &
 &[3pt] \ctrl{2}\gategroup[3,steps=2,style={densely dotted,inner sep=0pt}]{}
 & \ctrl{2} & & \\
\lstick{$\ket{+}$}
 & \gate{Z} & \ctrl{1}\gategroup[2,steps=2,style={densely dotted,inner sep=0pt}]{}
 & \ctrl{1} & & & & \\
\lstick{$\ket{0}$}
 & \gate{Z} & \targ{} & \targ{} & \targ{} & \targ{} & \gate{Z} &
\end{quantikz}}
 & \(=\) &
{\begin{quantikz}[row sep={0.68cm,between origins},column sep=0.22cm,align equals at=2]
\lstick{$\ket{+}$} & & \rstick{$\ket{+}$} \\
\lstick{$\ket{+}$} & \gate{Z} & \rstick{$\ket{-}$} \\
\lstick{$\ket{0}$} & & \rstick{$\ket{0}$}
\end{quantikz}}
\end{tabular}
\caption{Reduction of the toy circuit for \(u=01\). The signal supplies \(Z\) in slots \(01\) and \(11\). The arrows show the first \(Z\) being pushed left through the preceding CNOT, creating a \(Z\) on memory qubit 2. Each dotted CNOT pair then cancels, and the two remaining sensor \(Z\) operators cancel. Only \(Z\) on memory qubit 2 remains, giving \(\ket{+-}=\ket{01}_X\), the desired message string.}
\label{fig:pushing_Z}
\end{figure*}
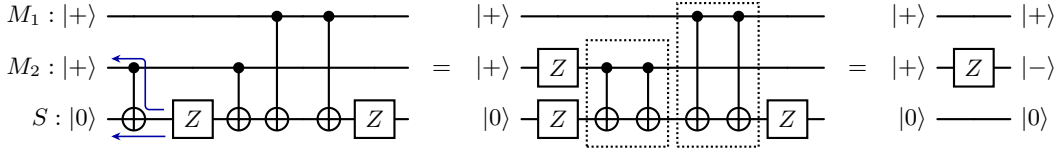

An alternative way to understand how this circuit works will become insightful later on. Write the initial memory state as \(\ket{+}^{\otimes k}=2^{-k/2}\sum_{z\in\mathbb{F}_2^k}\ket{z}\). For each basis state $\ket{z}$ of the memory and at each time slot \(x\), we can imagine that the control of the first CNOT computes the Boolean \emph{activation function}
\begin{equation}
f_z(x)=\begin{cases}
z_i,&x=e_i,\\
0,&\text{otherwise},
\end{cases}
\end{equation}
into the sensor, taking \(\ket{z}\ket{0}\) to \(\ket{z}\ket{f_z(x)}\) and possibly entangling the registers for superpositions of memory basis states. The signal then contributes a phase \(e^{i\theta f_z(x)c_u(x)}\) to the state \(\ket{z}\ket{f_z(x)}\), and the second control guarantees the sensor is restored to \(\ket{0}\) while the phase remains imparted on the memory. This process is called \textit{phase kickback}: after each slot, the sensor is disentangled and reset, while the memory retains the acquired relative phases. At each time step, we repeat the procedure of computing an activation function, conditionally activating the sensor qubit on the result, accumulating phase from a signal bit, and turning the sensor back off again. The final state is
\begin{equation}
\ket{\psi(c_u)}=
2^{-k/2}\sum_z
\exp\!\left(i\theta\sum_{x\ne 0}f_z(x)c_u(x)\right)\ket{z}.
\end{equation}
For this activation function, \(\sum_{x\ne 0}f_z(x)c_u(x)=\sum_i z_i u_i\), so at \(\theta=\pi\) we recover
\begin{equation}
\ket{\psi(c_u)}=2^{-k/2}\sum_z(-1)^{u\cdot z}\ket{z}=\ket{u}_X.
\end{equation}
For each \(\ket{z}\), the values of the activation function across all time slots form a bitstring \(f_z=(f_z(x))_{x\ne0}\), which we call a \emph{mask}. We can picture this mask overlaid on the incoming codeword: a phase is acquired at a position only when the corresponding bits of both strings equal one. The total phase picked up by the memory state $\ket{z}$ depends on the ordinary integer dot product of its mask and the codeword. Hence, to design an effective decoder, one must engineer an appropriate family of masks $f_z$ such that the bits of the message $u$ can be read out from the final state of the memory register. 

This toy protocol has several limitations. Firstly, it only distinguishes all messages with zero error when \(\theta=\pi\), where a receiver with only classical memory can already read each incoming bit without error. Indeed, memory qubit \(i\) could in principle be measured immediately after slot \(e_i\) to read out the $i$th message bit, so this example does not directly establish an advantage from \textit{persistent} memory. Furthermore, the protocol does not use the redundancy of the code: memory qubit \(i\) acquires phase only from position \(e_i\), while the remaining codeword positions contribute no phase to the memory. For larger codes, we could instead choose masks which overlap the possible incoming codewords at more positions, allowing more of the available phase to be accumulated in the memory. We will devise more sophisticated masks which allow perfect decoding to be achieved at weaker signal strengths.

\subsection{Simplex QQ decoder construction} 

We now generalize the toy decoder into a framework which will overcome these limitations and allow larger codes to be decoded exactly, ultimately even with a bounded number of transmission errors, at total strengths \(\Phi\) which preclude perfect CC or QC decoding. For now, however, we continue the assumption that there are no transmission errors, so that the received bitstring is an uncorrupted codeword. The circuit in Fig.~\ref{fig:general-circuit} depicts this generalized circuit, and uses a memory register of \(k\) qubits and a sensor register of \(S\) qubits (with possibly $S>1$). As in the toy example, for each memory basis state \(\ket{z}\) and time step \(x\), we evaluate the activation function \(f_z(x)\) and activate the sensor if \(f_z(x)=1\). Here, activating the sensor flips all \(S\) qubits from \(\ket{0}^{\otimes S}\) to \(\ket{1}^{\otimes S}\). Writing \(r(x)=c_u(x)\) for the received bit, each activated qubit acquires a phase \(\theta r(x)\), and the phases add constructively such that the total phase is \(S\theta r(x)\); we call \(S\theta\) the \emph{interaction strength}. Applying the same control again resets the sensor, giving the net action
\begin{equation}
\label{eq:general-phase-kickback}
\ket{z}\ket{0}^{\otimes S}
\longmapsto
e^{iS\theta r(x) f_z(x)}\ket{z}\ket{0}^{\otimes S}.
\end{equation}
The constructive interference over all \(S\) sensor qubits thus allows more signal to be accumulated per time slot.

\begin{figure*}[t]
\centering
\input{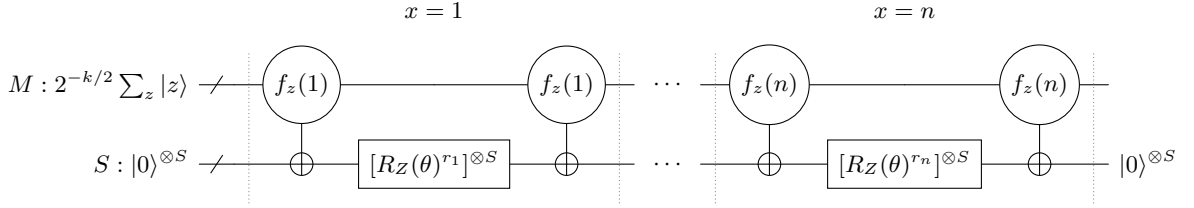}%
\caption{General simplex QQ decoder with \(k\) memory qubits and \(S\) sensor qubits. Each circle denotes a Boolean control: on memory basis state \(\ket{z}\), it flips every sensor qubit if \(f_z(x)=1\). The signal applies \(R_Z(\theta)^{r(x)}\) to each sensor qubit, where \(r(x)\) is the received bit. Repeating the control restores the sensor to \(\ket{0}^{\otimes S}\) and retains the phase \(e^{iS\theta r(x) f_z(x)}\) on the memory state $\ket{z}$.}
\label{fig:general-circuit}
\end{figure*}

But the most significant difference from the toy decoder is our choice of sensor activation function. At each time slot \(x\), we assign the activation function
\begin{equation}
\label{eq:affine-activation}
f_z(x)=x\cdot z\oplus b(x),
\end{equation}
where \(x\cdot z\) is evaluated modulo two and \(b(x)\in\{0,1\}\) is independent of the memory label \(z\). Thus, the masks can be written as \(f_z=c_z\oplus b\), with \(b=(b(x))_{x\ne0}\) the bitstring formed by $b(x)$ across all values of nonzero $x$. Equivalently, the masks form a coset of the simplex code. This structure is useful for three reasons. First, in the absence of transmission errors, the phase accumulated by \(\ket{z}\) is proportional to the ordinary integer dot product \(f_z\cdot c_u\). These masks can overlap the incoming codeword at many positions: for example, if \(b=0\), then at $z=u$, the mask overlaps the codeword at every position, since \(f_u=c_u\). It follows that every 1 in the codeword contributes to the phase of the memory state \(\ket{u}\). With a suitable choice of \(b\), this phase accumulation allows the decoder to succeed using weaker individual pulses. Second, the shared algebraic structure of the masks and codewords makes their overlaps, and hence the acquired phases, easier to analyze. Third, these activation functions are easy to implement using CNOT and single-qubit \(X\) gates, as we now show.

Figure~\ref{fig:affine_activation} shows the implementation of one decoder time step using elementary quantum gates, choosing \(k=4\) and \(x=1010\) as an example. First, we must compute the activation function $f_z(x) = (x\cdot z) \oplus b(x)$ and flip the sensor qubits if it evaluates to 1. We compute this function in two steps: the CNOT gates compute $(x\cdot z)$, and a possible $X$ gate on the sensor subsequently adds the offset bit $b(x)$ if it equals 1. Concretely, given any memory state \(\ket{z}\), the CNOTs from qubits 1 and 3 map each sensor qubit from \(\ket{0}\) to \(\ket{x\cdot z}=\ket{z_1\oplus z_3}\). Afterward, we apply an extra $X$ gate to each sensor qubit if \(b(x)=1\), thereby putting each into the state \(\ket{f_z(x)}=\ket{z_1\oplus z_3\oplus b(x)}\), as desired. After the signal pulse, we reverse these operations to reset the sensor while retaining the phase on the memory. For any other general \(x\), the CNOTs are instead controlled from the memory qubits $i$ for which \(x_i=1\). Notably, the control uses only Clifford gates, and the sensor register is returned to \(|0\rangle^{\otimes S}\), disentangled from the memory, after each time step. Thus, the sensor need only preserve coherence briefly, within individual time steps.

\begin{figure*}[t]
\centering
\small
\begin{tabular}{@{}c@{\hspace{8pt}}c@{\hspace{8pt}}c@{}}
\begin{quantikz}[thin lines,wire types={q,q},row sep={1.375cm,between origins},column sep=0.22cm,align equals at=1.5]
\lstick{$M$}
 &[6pt] \qwbundle{}
 & \gate[style={circle,inner sep=-1.5pt}]{f_z(x)}\wire[d][1]{q}
 &
 & \gate[style={circle,inner sep=-1.5pt}]{f_z(x)}\wire[d][1]{q}
 & \\
\lstick{$S$}
 & \qwbundle{}
 & \targ{}
 & \gate{R_Z(\theta)^{r_x}}
 & \targ{}
 &
\end{quantikz}
& \(=\) &
\begin{quantikz}[thin lines,row sep={0.55cm,between origins},column sep=0.17cm,align equals at=3.75]
\lstick{$M_1$}
 & & \ctrl{4} & & & & & & \ctrl{4} & \\
\lstick{$M_2$}
 & & & & & & & & & \\
\lstick{$M_3$}
 & & & \ctrl{2} & & & & \ctrl{2} & & \\
\lstick{$M_4$}
 & & & & & & & & & \\
\lstick{$S$}
 &[6pt] \qwbundle{}
 & \targ{}
 & \targ{}
 & \gate{X^{b(x)}}
 & \gate{R_Z(\theta)^{r_x}}
 & \gate{X^{b(x)}}
 & \targ{}
 & \targ{}
 &
\end{quantikz}
\end{tabular}
\caption{Implementation of one time step of the decoder in Fig.~\ref{fig:general-circuit}, shown for \(k=4\) and \(x=1010\), so that \(f_z(x)=z_1\oplus z_3\oplus b(x)\). On the right, each CNOT into the sensor register represents one CNOT to each of its \(S\) qubits.}
\label{fig:affine_activation}
\end{figure*}
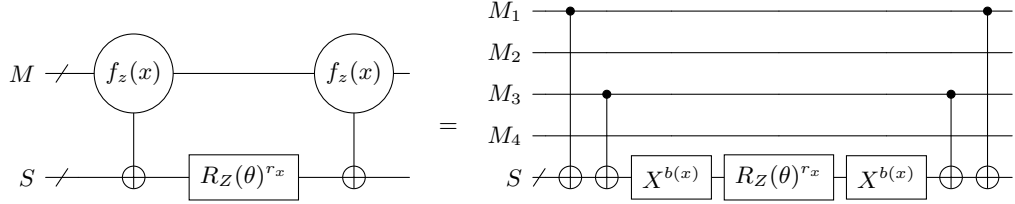

In the toy example, the circuit transduced the temporal pattern of codeword bits into a spatial pattern of single-qubit gates on the memory. Whether this is possible in the general decoder is less obvious. However, the shared simplex structure of the codewords and masks allows us to write the final memory state in a form which nonetheless makes the latent message pattern explicit. In the absence of transmission errors, and up to a global phase, we have
\begin{equation}
\label{eq:memory-message-pattern}
\ket{\psi(c_u)}=D_b X[u]D_b^\dagger\ket{+}^{\otimes k},
\end{equation}
where \(X[u]=\bigotimes_{i=1}^k X^{u_i}\) and \(D_b\) is a diagonal unitary determined by the offset \(b\) and the interaction strength, $S\theta$. Note that this is a purely formal decomposition: we never physically apply \(D_b\) or \(X[u]\), but use these operators to give an intuitive description of the circuit's total action. The bits of \(u\) thus appear as a spatial pattern of Pauli \(X\) gates; however, this pattern is potentially obfuscated by the unitary $D_b$ which conjugates it. Our goal is to choose \(b\) so that conjugation by \(D_b\) does not scramble the spatial pattern in a way that makes it unreadable by the final measurement.

We derive this decomposition as follows. Recall that the phase accumulated on a memory basis state \(\ket{z}\) is proportional to the number of positions at which the mask \(f_z\) and the codeword \(c_u\) both share 1s, given by $\sum_{x\ne0}f_z(x)c_u(x)$. Thus, the total unitary operation $U_{\rm tot}$ implemented by the circuit on a memory state $\ket{z}$ is
\begin{equation}
\label{eq:memory-overlap-action}
U_{\rm tot}\ket{z}
=\exp\left(iS\theta\sum_{x\ne0}f_z(x)c_u(x)\right)\ket{z}.
\end{equation}
We can count the number of common ones by adding the Hamming weights of the codeword and mask and subtracting the number of places where their bitstrings disagree. Thus
\begin{equation}
\label{eq:mask-codeword-overlap}
\sum_{x\ne0}f_z(x)c_u(x)
=\frac{|c_u|+|f_z|-|f_z\oplus c_u|}{2}.
\end{equation}
The shared simplex structure now simplifies the term $|f_z\oplus c_u|$. Because the simplex code is linear, we can write
\begin{equation}
\label{eq:mask-codeword-shift}
\begin{aligned}
f_z\oplus c_u
&=c_z\oplus c_u\oplus b\\
&=c_{z\oplus u}\oplus b=f_{z\oplus u}.
\end{aligned}
\end{equation}
Therefore, $U_{\rm tot}$ acts on a memory state \(\ket{z}\) by applying the following product of three phases:
\begin{equation}
\label{eq:memory-three-phases}
\begin{aligned}
U_{\rm tot}\ket{z}
={}&e^{iS\theta|c_u|/2}e^{iS\theta|f_z|/2}e^{-iS\theta|f_{z\oplus u}|/2}\ket{z}.
\end{aligned}
\end{equation}
The first phase is independent of \(z\) and hence is a global phase on the memory register. The other two depend on the weights of \(f_z\) and \(f_{z\oplus u}\), with opposite signs in their exponents. We now show how sequentially applying $X[u]$ and an appropriately defined diagonal unitary $D_b$ to a memory state $\ket{z}$ can recreate the action of $U_{\rm tot}$ on $\ket{z}$, up to a global phase. In particular, define the diagonal unitary \(D_b\) by
\begin{equation}
\label{eq:mask-diagonal-unitary}
D_b\ket{z}=e^{iS\theta|f_z|/2}\ket{z}.
\end{equation}
Now start with a bare memory state \(\ket{z}\). Applying \(X[u]\) once shifts the state to \(\ket{z\oplus u}\). Next, applying \(D_b^\dagger\) contributes the phase \(e^{-iS\theta|f_{z\oplus u}|/2}\), which matches the third phase term to appear in Eq. (\ref{eq:memory-three-phases}). A second \(X[u]\) then returns the state to \(\ket{z}\), after which another \(D_b\) contributes the phase \(e^{iS\theta|f_z|/2}\), which matches the second phase in Eq. (\ref{eq:memory-three-phases}). Thus, up to the global phase $e^{iS\theta|c_u|/2}$,
\begin{equation}
\label{eq:memory-unitary-decomposition}
U_{\rm tot}\ket{z}=D_b X[u]D_b^\dagger X[u]\ket{z}.
\end{equation}
Furthermore, since the memory register is initially prepared in \(\ket{+}^{\otimes k}\), and \(X[u]\ket{+}^{\otimes k}=\ket{+}^{\otimes k}\), the rightmost \(X[u]\) has no effect on this initial state. We thus recover Eq.~\eqref{eq:memory-message-pattern}, up to a global phase.

Equipped with this formal decomposition, we now return to choosing \(b\) so that \(D_b\) does not scramble the spatial pattern of \(X[u]\) in a way that makes it unreadable. Suppose, for a moment, that we could choose \(D_b=H^{\otimes k}\). Then the final memory state would equal 
\begin{equation}
\label{eq:hadamard-ideal-conjugation}
U_{\rm tot}\ket{+}^{\otimes k} = H^{\otimes k}X[u]H^{\otimes k}\ket{+}^{\otimes k}=Z[u]\ket{+}^{\otimes k} = \ket{u}_X,
\end{equation}
exactly the output we want. Unfortunately, however, \(D_b\) must be diagonal, whereas \(H^{\otimes k}\) is not. We must therefore instead find a diagonal \(D_b\) such that $D_b X[u] D_b^\dagger$ produces a readable spatial pattern of Paulis encoding the message.

Interestingly, a suitable \(D_b\) could in principle be constructed from controlled-\(Z\) (CZ) gates. Figure~\ref{fig:decomp} illustrates the choice \(D_b=\mathrm{CZ}_{12}\mathrm{CZ}_{34}\) for \(k=4\) and \(u=0110\), so that \(X[u]=X_2X_3\). Using the identity in Fig.~\ref{fig:circuit_identities}(b), we push \(X_2\) through the CZ on qubits 1 and 2, creating a new \(Z_1\). Similarly, pushing \(X_3\) through the CZ on qubits 3 and 4 creates a new \(Z_4\). The matching CZ gates now cancel in pairs, leaving \(Z_1X_2X_3Z_4\) on the memory. Since \(X\ket{+}=\ket{+}\), the remaining \(X\) operators are absorbed into the initial memory state, leaving \(Z_1Z_4\ket{+}^{\otimes4}=\ket{1001}_X\). Note that this final state does not equal the exact message state $\ket{u}_X = \ket{u_1u_2u_3u_4}_X$, as one might initially expect; instead, it equals a permutation of this state, namely $\ket{u_2u_1u_4u_3}_X$. Nonetheless, the original message can be read out unambiguously from this final state, since the permutation is invertible and can be easily undone. Therefore, if we can choose $b$ and the interaction strength $S\theta$ such that $D_b$ realizes this $\mathrm{CZ}$ circuit, then unambiguous decoding would be possible. More generally, for larger \textit{even} integer values of \(k\), the corresponding choice of $\mathrm{CZ}$ circuit would be
\begin{equation}
\label{eq:paired-cz-diagonal}
D_b=\mathrm{CZ}_{1,2}\mathrm{CZ}_{3,4}\cdots\mathrm{CZ}_{k-1,k},
\end{equation}
and the final state would be \(\ket{Pu}_X\) for any message \(u\), up to a global phase, where \(P\) is the \(k\times k\) permutation matrix
\begin{equation}
\label{eq:paired-bit-permutation}
P=I_{k/2}\otimes\begin{pmatrix}0&1\\1&0\end{pmatrix}
=\begin{pmatrix}
0&1&&&\\
1&0&&&\\
&&\ddots&&\\
&&&0&1\\
&&&1&0
\end{pmatrix}.
\end{equation}

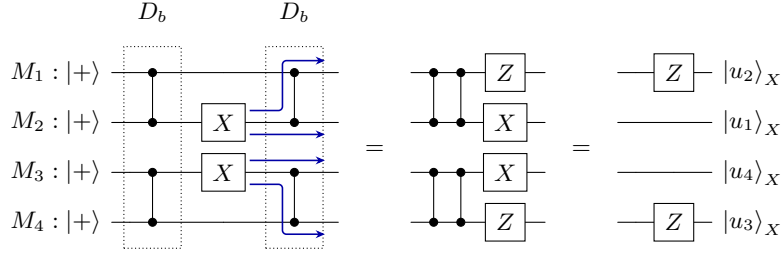
\begin{figure*}[t]
\centering
\small
\usetikzlibrary{arrows.meta}
\begin{tabular}{@{}c@{\hspace{7pt}}c@{\hspace{7pt}}c@{\hspace{7pt}}c@{\hspace{7pt}}c@{}}
{\begin{quantikz}[
 thin lines,row sep={0.66cm,between origins},column sep=0.24cm,align equals at=2.5,
 execute at end picture={
  \coordinate (pushxtwo) at ([xshift=1.5pt,yshift=-4.5pt]group\tikzcdmatrixname-2-4.east);
  \draw[blue!65!black,line width=0.6pt,-{Stealth[length=3.5pt,width=3pt]}]
   (pushxtwo) -- ([xshift=11.5pt]\tikzcdmatrixname-2-5.center |- pushxtwo);
  \draw[blue!65!black,line width=0.6pt,rounded corners=1.5pt,-{Stealth[length=3.5pt,width=3pt]}]
   ([xshift=1.5pt,yshift=4.5pt]group\tikzcdmatrixname-2-4.east)
   -- ([xshift=-4.5pt,yshift=4.5pt]\tikzcdmatrixname-2-5.center)
   -- ([xshift=-4.5pt,yshift=4.5pt]\tikzcdmatrixname-1-5.center)
   -- ([xshift=11.5pt,yshift=4.5pt]\tikzcdmatrixname-1-5.center);
  \coordinate (pushxthree) at ([xshift=1.5pt,yshift=4.5pt]group\tikzcdmatrixname-3-4.east);
  \draw[blue!65!black,line width=0.6pt,-{Stealth[length=3.5pt,width=3pt]}]
   (pushxthree) -- ([xshift=11.5pt]\tikzcdmatrixname-3-5.center |- pushxthree);
  \draw[blue!65!black,line width=0.6pt,rounded corners=1.5pt,-{Stealth[length=3.5pt,width=3pt]}]
   ([xshift=1.5pt,yshift=-4.5pt]group\tikzcdmatrixname-3-4.east)
   -- ([xshift=-4.5pt,yshift=-4.5pt]\tikzcdmatrixname-3-5.center)
   -- ([xshift=-4.5pt,yshift=-4.5pt]\tikzcdmatrixname-4-5.center)
   -- ([xshift=11.5pt,yshift=-4.5pt]\tikzcdmatrixname-4-5.center);
 }]
\lstick{$M_1:\ket{+}$}
 & & \ctrl{1}\gategroup[4,steps=1,style={densely dotted,inner xsep=6pt,inner ysep=5pt},label style={label position=above,anchor=south}]{$D_b$}
 &[10pt] &[10pt] \ctrl{1}\gategroup[4,steps=1,style={densely dotted,inner xsep=6pt,inner ysep=5pt},label style={label position=above,anchor=south}]{$D_b$} &[8pt] \\
\lstick{$M_2:\ket{+}$}
 & & \control{} & \gate{X} & \control{} & \\
\lstick{$M_3:\ket{+}$}
 & & \ctrl{1} & \gate{X} & \ctrl{1} & \\
\lstick{$M_4:\ket{+}$}
 & & \control{} & & \control{} &
\end{quantikz}}
 & \(=\) &
{\begin{quantikz}[thin lines,row sep={0.66cm,between origins},column sep=0.24cm,align equals at=2.5]

 & \ctrl{1}
 & \ctrl{1} & \gate{Z} & \\

 & \control{} & \control{} & \gate{X} & \\

 & \ctrl{1}
 & \ctrl{1} & \gate{X} & \\

 & \control{} & \control{} & \gate{Z} &
\end{quantikz}}
 & \(=\) &
{\begin{quantikz}[thin lines,row sep={0.66cm,between origins},column sep=0.24cm,align equals at=2.5]
 & & \gate{Z} & \rstick{$\ket{u_2}_X$} \\
 & & & \rstick{$\ket{u_1}_X$} \\
 & & & \rstick{$\ket{u_4}_X$} \\
 & & \gate{Z} & \rstick{$\ket{u_3}_X$}
\end{quantikz}}
\end{tabular}
\caption{Formal decomposition of the total decoder unitary applied to the initial memory state \(\ket{+}^{\otimes4}\), for \(D_b=\mathrm{CZ}_{12}\mathrm{CZ}_{34}\) and \(u=0110\). The arrows show the \(X\) operators being pushed through the CZ gates and generating \(Z\) on the other qubit of each pair. Dotted boxes identify the two \(D_b\) layers. The leftover \(X\) operators act trivially on \(\ket{+}\), leaving \(\ket{Pu}_X=\ket{1001}_X\). Note that the displayed gates describe the circuit's net action; they are not additional physical operations.}
\label{fig:decomp}
\end{figure*}

However, it remains to choose \(b\) and a suitable interaction strength so that \(D_b\) realizes the desired CZ circuit. Given any choice of $b$, both $D_b$ and the CZ circuit are diagonal in the computational basis, so applying either operator to \(\ket{z}\) multiplies it by a phase factor. Our goal is to match these factors for every \(z\), up to a common global phase. We first examine the action of $D_b$ (as provided in Eq. (\ref{eq:mask-diagonal-unitary})) on a memory state $\ket{z}$:
\begin{equation}
\label{eq:offset-diagonal-action}
D_b\ket{z}=\exp\!\left[\frac{iS\theta}{2}\sum_{x\ne0}\bigl(b(x)\oplus x\cdot z\bigr)\right]\ket{z}.
\end{equation}
We will next compute the phase factor obtained by applying the CZ circuit to the same \(\ket{z}\); thereafter, we will rewrite this factor in the same form as the above exponential, so that comparing the two expressions identifies the \(b\) and interaction strength $S\theta$ required for $D_b$ to equal the CZ circuit. Thus observe that for a single CZ gate, \(\mathrm{CZ}_{1,2}\ket{z_1z_2}=(-1)^{z_1z_2}\ket{z_1z_2}\), since CZ contributes a minus sign exactly when both bits are 1. The full CZ circuit therefore maps \(\ket{z}\) to \((-1)^{q(z)}\ket{z}\), where
\begin{equation}
\label{eq:paired-quadratic-function}
q(z)=z_1z_2\oplus z_3z_4\oplus\cdots\oplus z_{k-1}z_k
\end{equation}
is a quadratic Boolean function. 

Fortunately, it turns out that this function has a special and very useful property: \((-1)^{q(z)}\) equals its own normalized Walsh--Hadamard transform,
\begin{equation}
\label{eq:quadratic-fourier-identity}
(-1)^{q(z)}=2^{-k/2}\sum_{x\in\mathbb{F}_2^k}(-1)^{q(x)\oplus x\cdot z}.
\end{equation}
The left side is the phase factor acquired by \(\ket{z}\) under the full CZ circuit, while the right side contains \(q(x)\oplus x\cdot z\), which has the same form as the \(b(x)\oplus x\cdot z\) in the exponent defining \(D_b\) in Eq.~\eqref{eq:offset-diagonal-action}. This correspondence suggests that choosing \(b(x)=q(x)\) may make $D_b$ equal to the CZ circuit. To verify this, apply the identity \((-1)^t=1-2t\) (which holds for any bit \(t\in\{0,1\}\)), to both sides of the equation above:
\begin{equation}
\label{eq:quadratic-fourier-bits}
\begin{aligned}
1-2q(z)
&=2^{-k/2}\sum_{x\in\mathbb{F}_2^k}\left[1-2\bigl(q(x)\oplus x\cdot z\bigr)\right]\\
&=2^{k/2}-2^{1-k/2}\sum_{x\ne0}\bigl(q(x)\oplus x\cdot z\bigr).
\end{aligned}
\end{equation}
Here the \(x=0\) term in the bit sum vanishes because \(q(0)=0\). Rearranging gives
\begin{equation}
\label{eq:quadratic-mask-weight}
q(z)=\frac{1-2^{k/2}}{2}+2^{-k/2}\sum_{x\ne0}\bigl(q(x)\oplus x\cdot z\bigr).
\end{equation}
Thus, writing \((-1)^{q(z)}=e^{i\pi q(z)}\), the phase factor acquired by \(\ket{z}\) under the CZ circuit becomes
\begin{equation}
\label{eq:quadratic-cz-phase}
\begin{aligned}
(-1)^{q(z)}
={}&e^{i\pi(1-2^{k/2})/2}\exp\!\left[\frac{i\pi}{2^{k/2}}\sum_{x\ne0}\bigl(q(x)\oplus x\cdot z\bigr)\right].
\end{aligned}
\end{equation}
The first factor is independent of \(z\). Matching the second to the phase factor produced by \(D_b\) in Eq.~\eqref{eq:offset-diagonal-action} gives
\begin{equation}
\label{eq:quadratic-decoder-choice}
b(x)=q(x),\qquad \theta=\frac{2\pi}{S\,2^{k/2}}.
\end{equation}
With this choice, \(D_b\) equals the desired CZ circuit up to a global phase. For fixed \(S\), the interaction strength \(S\theta\) tends to zero as \(k\) grows.

We have therefore shown that, for an \([n,k]\) simplex code with even \(k\), this QQ decoder can perfectly identify every message in the absence of transmission errors. We use \(S\) sensing qubits and \(k\) memory qubits, for a total of \(Q=k+S\) qubits. The total strength, that is, the signal strength $\theta$ times the \textit{total} number of qubits available, is then
\begin{equation}
\label{eq:qq-total-strength}
\Phi=(k+S)\theta=\frac{2\pi(k+S)}{S2^{k/2}}=\frac{2\pi(k+S)}{S\sqrt{n+1}}.
\end{equation}
For even \(k\ge6\) and choosing $S=k$, this value of $\Phi$ is below \(\pi\), in the regime where perfect CC and QC decoding are impossible. Thus, the generalized circuit makes better use of the phase available across successive signal pulses. As the code grows, there are more time slots over which phase can accumulate, allowing both the signal strength and the total strength to decrease while the total number of qubits grows only logarithmically with \(n\). 

Figure~\ref{fig:mask-overlap} gives a fully worked example of the decoder instance depicted in Figure~\ref{fig:decomp}. Starting from the message \(u=0110\), it shows its encoding as a codeword $c_u$ and also lists each mask $f_z$ associated with a memory basis state \(\lvert z\rangle\). For a particular $\ket{z}$, the number \(N_z = \sum_xf_z(x)c_u(x)\) of positions where both the mask and codeword equal 1 determines the phase accumulated by that basis state. Together, these phases produce the sign pattern of the desired output \(\lvert Pu\rangle_X=\lvert1001\rangle_X\).
\begin{figure}[t]
\centering

\providecommand{\maskOverlapCaption}{}
\renewcommand{\maskOverlapCaption}{%
Masks for the $[15,4]$ simplex code with message $u=0110$ and offset
$b(x)=x_1x_2\oplus x_3x_4$. \textbf{(a)} Columns are time slots whose labels form the
generator matrix; $c_u$ is the sum of the highlighted rows. \textbf{(b)} Row $\ket{z}$
shows the mask $f_z=c_z\oplus b$: shaded cells activate the sensor, and
teal cells also coincide with a signal pulse. The count
$N_z=\sum_x f_z(x)c_u(x)$ is the ordinary integer dot product of the mask $f_z$ and codeword $c_u$. Each teal
cell adds phase $S\theta=\pi/2$ to $\ket{z}$, so
$s_z=e^{iS\theta N_z}$ is the sign acquired by $\ket{z}$.
These signs produce the final memory state $\ket{Pu}_X=\ket{1001}_X$.}

\begin{tikzpicture}[
  x=0.40cm,
  y=-0.40cm,
  scale = 0.95,
  font=\scriptsize,
  every node/.style={inner sep=0pt},
  cell/.style={draw=black!45,line width=0.3pt,fill=white},
  one/.style={draw=mskTealD,line width=0.45pt,fill=mskTeal},
  off/.style={draw=black!12,line width=0.3pt,fill=white},
  on/.style={draw=mskPurpD,line width=0.5pt,fill=mskPurp!70},
  hit/.style={draw=mskTealD,line width=0.45pt,fill=mskTeal},
  offs/.style={draw=black!50,line width=0.45pt,fill=black!14},
  lab/.style={anchor=east,text=black!80}
]

\definecolor{mskTeal}{HTML}{CFEDE6}
\definecolor{mskTealD}{HTML}{1F7A63}
\definecolor{mskPurp}{HTML}{E6DEF5}
\definecolor{mskPurpD}{HTML}{5B3F9A}

\def\mskU{6}    
\def\mskC{0.84} 

\def\mskbit#1#2#3{%
  \pgfmathtruncatemacro#1{mod(floor(#2/2^(4-#3)),2)}%
}

\node[anchor=south west,font=\small\bfseries]
  at (-3.2,0.2) {(a)};
\foreach \mski in {1,...,4}{
  \mskbit\mskui\mskU\mski

  \ifnum\mskui=1
    \fill[mskTeal!55]
      (0.5,\mski-0.5) rectangle (15.5,\mski+0.5);
    \def\mskst{one}
  \else
    \def\mskst{cell}
  \fi

  \draw[\mskst]
    (-1.5-\mskC/2,\mski-\mskC/2)
    rectangle ++(\mskC,\mskC);
  \node at (-1.5,\mski) {\mskui};

  \foreach \mskx in {1,...,15}{
    \mskbit\mskxi\mskx\mski
    \node[text=black!55] at (\mskx,\mski) {\mskxi};
  }
}

\node[anchor=south] at (-1.5,0.2) {$u$};
\node[anchor=south,text=black!70] at (8,0.2) {time slot $x$};

\foreach \mskx in {1,...,15}{
  \mskbit\mxa\mskx1
  \mskbit\mxb\mskx2
  \mskbit\mxc\mskx3
  \mskbit\mxd\mskx4

  \mskbit\mskua\mskU1
  \mskbit\mskub\mskU2
  \mskbit\mskuc\mskU3
  \mskbit\mskud\mskU4

  \pgfmathtruncatemacro\mskcu{
    mod(\mskua*\mxa+\mskub*\mxb+\mskuc*\mxc+\mskud*\mxd,2)
  }
  \pgfmathtruncatemacro\mskbx{
    mod(\mxa*\mxb+\mxc*\mxd,2)
  }

  \ifnum\mskcu=1
    \def\mskst{one}
  \else
    \def\mskst{cell}
  \fi

  \draw[\mskst]
    (\mskx-\mskC/2,5.6-\mskC/2)
    rectangle ++(\mskC,\mskC);
  \node at (\mskx,5.6) {\mskcu};

  \ifnum\mskbx=1
    \def\mskst{offs}
  \else
    \def\mskst{cell}
  \fi

  \draw[\mskst]
    (\mskx-\mskC/2,6.7-\mskC/2)
    rectangle ++(\mskC,\mskC);
  \node at (\mskx,6.7) {\mskbx};
}

\node[lab] at (0.25,5.6) {$c_u$};
\node[lab] at (0.25,6.7) {$b$};

\begin{scope}[shift={(0,2.0)}]

\node[anchor=south west,font=\small\bfseries]
  at (-3.2,8.45) {(b)};

\node[anchor=south,text=black!70] at (8,8.45)
  {mask $f_z=c_z\oplus b$};
\node[anchor=south] at (16.55,8.45) {$N_z$};
\node[anchor=south] at (17.75,8.45) {$s_z$};

\node[rotate=90,text=black!70] at (-3.2,15.47)
  {memory basis states $\ket{z}$};

\foreach \mskz in {0,...,15}{
  \pgfmathsetmacro\msky{8.75+\mskz*0.84+0.42}
  \gdef\mskN{0}

  \mskbit\mskza\mskz1
  \mskbit\mskzb\mskz2
  \mskbit\mskzc\mskz3
  \mskbit\mskzd\mskz4

  \node[lab,text=mskPurpD] at (0.25,\msky)
    {$\ket{\mskza\mskzb\mskzc\mskzd}$};

  \foreach \mskx in {1,...,15}{
    \mskbit\mxa\mskx1
    \mskbit\mxb\mskx2
    \mskbit\mxc\mskx3
    \mskbit\mxd\mskx4

    \mskbit\mskua\mskU1
    \mskbit\mskub\mskU2
    \mskbit\mskuc\mskU3
    \mskbit\mskud\mskU4

    \pgfmathtruncatemacro\mskcu{
      mod(\mskua*\mxa+\mskub*\mxb+\mskuc*\mxc+\mskud*\mxd,2)
    }
    \pgfmathtruncatemacro\mskf{
      mod(
        \mskza*\mxa+\mskzb*\mxb+\mskzc*\mxc+\mskzd*\mxd
        +\mxa*\mxb+\mxc*\mxd,
        2
      )
    }

    \ifnum\mskf=1
      \ifnum\mskcu=1
        \xdef\mskN{\the\numexpr\mskN+1\relax}
        \draw[hit]
          (\mskx-\mskC/2,\msky-\mskC*0.42)
          rectangle ++(\mskC,\mskC*0.84);
      \else
        \draw[on]
          (\mskx-\mskC/2,\msky-\mskC*0.42)
          rectangle ++(\mskC,\mskC*0.84);
      \fi
    \else
      \draw[off]
        (\mskx-\mskC/2,\msky-\mskC*0.42)
        rectangle ++(\mskC,\mskC*0.84);
    \fi
  }

  \node at (16.55,\msky) {\mskN};

  \pgfmathtruncatemacro\msks{mod(\mskN,4)}
  \ifnum\msks=0
    \node[text=mskTealD] at (17.75,\msky) {$+$};
  \else
    \node[text=mskPurpD] at (17.75,\msky) {$-$};
  \fi
}

\draw[black!45,line width=0.4pt]
  (18.3,8.82) -- (18.5,8.82)
  -- (18.5,22.13) -- (18.3,22.13);

\node[anchor=north] at (7.5,23.0) {$
  \underbrace{
    \sum_z\ket{z}
  }_{\substack{\text{initial}\\\text{state}}}
  \;\xrightarrow{\text{decoder}}\;
  \underbrace{
    \sum_z s_z\ket{z}
    =\sum_z(-1)^{Pu\cdot z}\ket{z}
    =\ket{Pu}_X
    =\ket{1001}_X
  }_{\text{\makebox[0pt]{final memory state}}}
$};

\end{scope}
\end{tikzpicture}%
\caption{\maskOverlapCaption}
\label{fig:mask-overlap}
\end{figure}

An interesting feature of this decoding scheme is that the memory outputs the bitstring \(\ket{Pu}_X\), rather than the true message string \(\ket{u}_X\). However, since \(P\) is invertible, the measured string \(Pu\) still allows us to recover \(u\) exactly. One might ask whether other choices of \(P\) can be realized through a suitable choice of \(b\), allowing equivalent decoding at the same total strength \(\Phi\). This is indeed possible, and will be useful when we discuss how to handle transmission errors. The following lemma formalizes this statement.

\begin{lemma}
\label{lem:general-quadratic-decoder}
Let \(k\) be an even integer and let \(P\in\mathbb{F}_2^{k\times k}\) be invertible and symmetric, with zero diagonal. Define
\begin{equation}
\label{eq:general-quadratic-offset}
b(x)=\bigoplus_{1\le i<j\le k}(P^{-1})_{ij}x_i x_j,
\end{equation}
for every \(x\in\mathbb{F}_2^k\), where the inverse of $P$ is evaluated over \(\mathbb{F}_2\). Then for the \([n,k]\) simplex code with signal strength \(\theta=2\pi/(S2^{k/2})\), the final state of the memory register of the above QQ decoder is \(\ket{Pu}_X\) (up to a global phase) for every message \(u\in\mathbb{F}_2^k\) in the absence of transmission errors.
\end{lemma}
\begin{proof}
See Appendix~\ref{app:general-quadratic-decoder}.
\end{proof}

We note an intriguing connection between this QQ decoder and \emph{binary bent functions}. Every \(b\) given by Lemma~\ref{lem:general-quadratic-decoder} is bent, meaning that the Walsh--Hadamard transform of \((-1)^{b(x)}\) has constant magnitude:
\begin{equation}
\label{eq:bent-offset-definition}
\left|\sum_{x\in\mathbb{F}_2^k}(-1)^{b(x)\oplus x\cdot z}\right|=2^{k/2}
\end{equation}
for every \(z\in\mathbb{F}_2^k\). For example, Eq.~\eqref{eq:quadratic-fourier-identity} immediately shows that \(b(x)=q(x)\) in our worked example is bent. Bent functions are well studied in coding theory and cryptography, with applications to code construction and the design of cryptographic Boolean functions \cite{rothaus1976bent,carlet2021boolean}. 

\subsection{Simplex QQ decoding with transmission errors}\label{sec:trans-errors}
The QQ decoder above can perfectly identify uncorrupted simplex codewords at total strengths too weak for CC and QC protocols to do the same. We now extend this construction to preserve perfect decoding when up to a fixed number \(t\) of the incoming codeword bits are flipped. The first step is to understand how a transmission bit flip propagates through the decoder circuit. Its effect on the memory is simple: for an underlying message \(u\), a flip at position \(x'\) changes the final state from \(\ket{Pu}_X\) to a superposition of \(\ket{Pu}_X\) and \(\ket{Pu\oplus x'}_X\).

To understand this fact, suppose we receive a corrupted bitstring \(r\), equal to the clean codeword $c_u$ but with a single bit flip at position \(x'\). Letting \(U_{\rm err}\) and $U_{\rm tot}$ be the total unitary operations of the circuit on the memory register due to reception of the corrupted $r$ and clean $c_u$, respectively, we now find an expression for $U_{\rm err}U_{\rm tot}^\dagger$, which captures the component of the evolution due to the single bit flip:
\begin{equation}
  \underbrace{U_{\rm err}}_{\rm corrupted} = \underbrace{\left(U_{\rm err}U_{\rm tot}^\dagger\right)}_{\rm bitflip}\underbrace{U_{\rm tot}}_{\rm clean}.
\end{equation} 
Recall Eq.~\eqref{eq:memory-overlap-action} gives the action of $U_{\rm tot}$ on a memory state $\ket{z}$; this equation also gives the action of $U_{\rm err}$, provided that we replace the clean codeword bits $c_u(x)$ with the corrupted bits $r(x)$. Thus, the action of $U_{\rm err}U_{\rm tot}^\dagger$ on a basis state $\ket{z}$ can be described as 
\begin{equation}
\label{eq:single-error-phase-difference}
\begin{aligned}
&U_{\rm err}U_{\rm tot}^\dagger\ket{z}\\
&\quad=\exp\left(iS\theta\sum_{x\ne0}f_z(x)[r(x)-c_u(x)]\right)\ket{z}\\
&\quad=\exp\left(iS\theta f_z(x')\bigl([1-c_u(x')]-c_u(x')\bigr)\right)\ket{z}\\
&\quad=\exp\left(iS\theta (-1)^{c_u(x')}f_z(x')\right)\ket{z}
\end{aligned}
\end{equation}
The last equality uses the fact that the bracketed bit difference is \(+1\) for a \(0\to1\) flip and \(-1\) for a \(1\to0\) flip. Now define \(Z[x']=\bigotimes_i Z^{x_i'}\), whose eigenvalue on \(\ket{z}\) is \((-1)^{x'\cdot z}\). Using \(f_z(x')=[1-(-1)^{b(x')}(-1)^{x'\cdot z}]/2\) and dropping the constant global phase gives
\begin{equation}
\label{eq:single-error-hamiltonian}
\begin{aligned}
U_{\rm err}U_{\rm tot}^{\dagger}
&=e^{-i(S\theta/2)H_{\rm err}},\\
H_{\rm err}&=(-1)^{b(x')\oplus c_u(x')}Z[x'].
\end{aligned}
\end{equation}
Because \(Z[x']\) changes an \(X\)-basis label by \(x'\), repeated applications of \(H_{\rm err}\) alternate between \(\ket{Pu}_X\) and \(\ket{Pu\oplus x'}_X\), up to a sign. Its exponential can therefore only produce a superposition of these two states.

Now that we know how transmission errors propagate through the circuit, we can discuss how to correct them. A single decoder cannot guarantee recovery of the message: its measured output may be shifted from \(Pu\), with no indication that a shift has occurred. The trick is to introduce redundancy by running \(B\) decoding circuits in parallel on the same received bitstring, using \(B\) memory registers of \(k\) qubits and \(B\) sensor registers of \(S\) qubits. However, simply repeating the same circuit produces identical copies of its output state. Although their measurement outcomes need not agree, and repetition can improve the measurement statistics, it does not guarantee that the underlying message can be identified without error.

Fortunately, Lemma~\ref{lem:general-quadratic-decoder} lets us choose a different transformation of the message for each decoder. We choose invertible, symmetric matrices \(P_1,\ldots,P_B\) with zero diagonal, and use the lemma to choose a corresponding offset function \(b_i(x)\) for each circuit. In the absence of transmission errors, decoder \(i\) then outputs \(\ket{P_i u}_X\), so that the parallel circuits encode the same message in different ways.

We now claim that, provided the number of transmission bit flips is bounded and the matrices \(P_i\) are chosen appropriately, the final measurements of the memory registers determine the underlying codeword unambiguously. The following lemma gives a sufficient condition.

\begin{lemma}
\label{lem:parallel-decoder-error-correction}
Let \(k\) be even, let \(t\ge0\) be an integer, and suppose \(B\ge2t+1\). Let \(P_1,\ldots,P_B\in\mathbb{F}_2^{k\times k}\) be invertible, symmetric matrices with zero diagonal satisfying
\begin{equation}
\label{eq:parallel-decoder-span-condition}
\dim_{\mathbb{F}_2}\operatorname{span}\{P_1w,\ldots,P_Bw\}>2t
\end{equation}
for every nonzero \(w\in\mathbb{F}_2^k\). Construct \(B\) parallel QQ decoders with these matrices according to Lemma~\ref{lem:general-quadratic-decoder}, at signal strength \(\theta=2\pi/(S2^{k/2})\). If all decoders receive the same simplex codeword with at most \(t\) flipped bits, their final \(X\)-basis measurement outcomes determine the transmitted message with zero error. With \(k\) memory qubits and \(S\) sensor qubits per decoder, this scheme uses \(Q=B(k+S)\) qubits and total strength
\begin{equation}
\label{eq:parallel-decoder-collective-strength}
\Phi=\frac{2\pi B(k+S)}{S2^{k/2}} = \frac{2\pi B(k+S)}{S\sqrt{n+1}}.
\end{equation}
\end{lemma}
We can now justify Lemma~\ref{lem:parallel-decoder-error-correction} by demonstrating that the following naive decoding protocol works: simply search through all \(2^k=n+1\) possible messages and ``test" which one fits the final $X$-basis measurement outcomes obtained from each memory register using the procedure we describe below. Although this brute-force search is exponential in \(k\), it is polynomial in \(n\) and \(B\); it serves to establish the lemma, and we revisit more efficient decoding schemes later. Suppose the memory measurements return the bitstrings \(y_1,\ldots,y_B\) and we propose a candidate message \(u\). Then since the ideal outputs should be \(P_1u,\ldots,P_Bu\), we can interpret our actual measurement results as $y_j = P_ju \oplus e_j$ for some implied error bitstrings $e_j = y_j \oplus P_ju$. Note that if \(u\) were correct and a single bit were flipped at \(x'\), each error $e_j$ would necessarily either be \(0\) or \(x'\), since each memory register $j$ would have been in a superposition of $\ket{P_j u}_X$ and $\ket{P_j u\oplus x'}_X$ prior to measurement. More generally, with up to $t$ flips, we could apply the single-flip argument successively to see that every error $e_j$ would necessarily be a binary sum of the flipped positions. Thus, the observed errors for the true message must span a subspace of dimension at most \(t\). We consequently accept a candidate \(u\) if
\begin{equation}
\label{eq:parallel-decoder-candidate-test}
\dim_{\mathbb{F}_2}\operatorname{span}\{y_1\oplus P_1u,\ldots,y_B\oplus P_Bu\}\le t,
\end{equation}
that is, if the errors implied by the candidate span a subspace of dimension $\leq t$.

The true message clearly passes this test, but we must show that no incorrect message can also pass. Suppose \(u\) is the true message and \(v\ne u\) is another candidate that passes. For every register \(j\), adding the two implied errors cancels the observed bitstring $y_j$:
\begin{equation}
\label{eq:parallel-decoder-candidate-difference}
(y_j\oplus P_ju)\oplus(y_j\oplus P_jv)=P_j(u\oplus v).
\end{equation}
Because both candidates pass, the implied errors for \(u\) can be spanned by at most \(t\) vectors per Eq. (\ref{eq:parallel-decoder-candidate-test}), and likewise for \(v\). Together, these two sets of at most \(t\) spanning vectors span every \(P_j(u\oplus v)\). The \(P_j(u\oplus v)\) therefore span a space of dimension at most \(2t\), contradicting Eq.~\eqref{eq:parallel-decoder-span-condition} for the nonzero string \(u\oplus v\). Thus the true message is the unique candidate that passes, proving Lemma~\ref{lem:parallel-decoder-error-correction}.

It remains to show that the required matrices \(P_j\) satisfying Eq.~\eqref{eq:parallel-decoder-span-condition} can actually be constructed. Related families of matrices with nonsingular pairwise differences arise in the study of Kerdock codes~\cite{kantor1992kerdockplanes}. A simple finite-field construction suffices. Write \(m=k/2\) and represent each element of \(\mathbb{F}_{2^m}\) by its \(m\) coordinates in a fixed binary basis. Multiplication by a field element \(a\) is then represented by an \(m\times m\) binary matrix \(M_a\). Since nonzero field elements have multiplicative inverses, \(M_a\) is invertible for \(a\ne0\); moreover, \(M_{a+b}=M_a+M_b\). Now define
\begin{equation}
\label{eq:finite-field-decoder-matrices}
P(a)=\begin{pmatrix}
0&M_a\\
M_a^{\mathsf T}&0
\end{pmatrix}.
\end{equation}
For \(a\ne0\), this matrix is invertible, symmetric, and has zero diagonal. Choose \(B=2t+1\) field elements \(a_1,\ldots,a_B\) linearly independent over \(\mathbb{F}_2\), which is possible when \(B\le m\), and take \(P_j=P(a_j)\). Any nonzero binary linear combination of the \(P_j\) equals \(P(a)\) for a nonzero linear combination \(a\) of the \(a_j\), and is therefore invertible. Consequently, for nonzero \(w\), no nonzero binary combination of \(P_1w,\ldots,P_Bw\) can vanish. Their span therefore has dimension \(B=2t+1>2t\), as required.

Finally, this finite-field construction permits more efficient classical decoding of the measurement readouts. After a fixed invertible change of binary coordinates, the ideal tuple \((P_1u,\ldots,P_Bu)\) has the form \((a_1\widetilde{u},\ldots,a_B\widetilde{u})\) over \(\mathbb{F}_{2^k}\), where \(\widetilde{u}\) represents the message and the \(a_j\) lie in its subfield \(\mathbb{F}_{2^m}\). These tuples form a length-\(B\), dimension-one Gabidulin code with minimum rank distance \(B\)~\cite{gabidulin1985rankdistance}. The errors measured have rank at most \(t\), which is preserved by the coordinate change. Standard Gabidulin decoders therefore recover the message from the readouts in time polynomial in \(k\), since \(B=2t+1\le k/2\)~\cite{silva2009fastgabidulin}.

\section{Quantum Advantage for Quantum Codeword Sensing}
\label{sec:theorems}

Having now established general bounds on the success probabilities of classical memory receivers, as well as provided a specific QQ protocol, we now show that there exist parameter regimes in which there is a separation between QQ and classical memory protocol success probability. In the first section, we prove that in the limit of large code size, and for specifically chosen signal strengths and qubit budgets, the success probability of a QQ receiver is either always one (for few enough transmission errors) or approaches one (for a constant fraction of adversarial transmission errors, or independent stochastic errors with small enough probability). For the parameter choices made in Theorem~\ref{thm:perfect-decoding} and Remark~\ref{rem:remark}, the success probability of a classical memory receiver tends to zero.

We then turn to analyzing the performance of receivers in the finite \(k\) regime. We find that using \(40\) qubits and with small enough transmission error probability, our QQ protocol for the simplex code can appreciably outperform a concrete QC receiver, and at \(80\) qubits, can outperform the bound given in Proposition~\ref{prop:qc-stochastic}, and thus any QC receiver.

\subsection{Asymptotic Advantage}
\label{sec:asymptotic}

We first begin by considering a regime where the QQ protocol succeeds with probability one.

\begin{theorem}
  \label{thm:perfect-decoding}
   Consider the transmission of a simplex code with \(k\) message bits, where \(k\) is even, in which at most \(t\) errors occur, for \(2t + 1 \le k/2\), and each message is sent with the same probability. Using \((2t+1)k\) memory qubits, \((2t+1)S\) sensor qubits, and \(\theta = \frac{2\pi}{S \cdot 2^{k/2}}\), a QQ receiver may determine the message with probability one. For the same number of qubits and \(\theta\), even in the absence of transmission errors, if \(S = \omega((t+1)\sqrt{k})\) and \(t = o(\sqrt{k})\), then the success probability of any QC receiver approaches zero, and if \(S = \omega(\sqrt{k})\), then the success probability of any CC receiver approaches zero.
\end{theorem}

\begin{proof}
  This follows directly from the application of Lemma \ref{lem:parallel-decoder-error-correction} with the explicit $P_j$ created by the finite field construction in Section \ref{sec:trans-errors} (which requires $B = 2t+1 \leq k/2$), and from Proposition~\ref{prop:qc-stochastic} with \(p=0\). More specifically, if we plug \(Q = (2t+1)(k+S)\) and \(\theta = 2\pi/(S2^{k/2})\) the small \(\Phi\) approximation from Equation~\eqref{eq:small-angle-qc-approx}, we find that the success probability vanishes when \((t+1)(1 + k/S) = o(\sqrt{k})\), and thus when t = \(o(\sqrt{k})\) and \(S = \omega((t+1)\sqrt{k})\). For the CC receiver, we have that for small \(\theta\), \(nC_{CC} \approx \frac{\pi^2}{2 \ln 2}\frac{(2t+1)(k+S)}{S^2}\), which vanishes as \(k \to \infty\) when \(S = \omega(\sqrt{k})\).
\end{proof}

Next, we show that under a fixed fraction of adversarial errors, our QQ receiver can still transmit information reliably in the limit of large code size.

\begin{theorem}
  \label{thm:adversarial-qq}
  For every fixed \(\tau < 1/16\), a QQ receiver using the binary simplex code with even \(k\), \(B\cdot k\) memory qubits, \(B \cdot S\) sensor qubits, and signal strength \(\theta = \frac{2\pi}{S 2^{k/2}}\) can correct any set of at most \(\tau n\) transmission errors with success tending to one as \(k \to \infty\), so long as \(B = B(k) = \omega(1)\).
\end{theorem}

\begin{proof}
  Let \(u\) be the transmitted message, and suppose transmission errors occur during time slots \(E \subseteq [n]\), where \(|E| \le \lfloor \tau n \rfloor\). Our decoding process will be simple. At the end of the protocol, we will measure our \(B\) registers in the \(X\) basis, apply \(P_{i}^{-1}\) to each measured bit string, and take a majority vote of the recovered values. If no result obtains a majority, we declare a decoding failure.

  If \(\ket{\psi(c_u)}_i\) is the state of a memory register \(i\) at the end of the protocol prior to the \(X\) basis measurement in the absence of errors, then after the errors in set \(E\), the state is
  \begin{equation}
    e^{-i S\theta H_{E,i}/2} \ket{\psi(c_u)}_i
  \end{equation}
  where
  \begin{equation}
     H_{E,i} = \sum_{x \in E} (-1)^{b_i(x)} (-1)^{c_u(x)} Z[x]
  \end{equation}
  and \(S \theta /2 = \pi \cdot 2^{-k/2}\). For any \(x \ne 0\),
  \begin{equation}
  _i\langle \psi(c_u) | Z[x] |\psi(c_u)\rangle_i = 0.
\end{equation}
so
\begin{equation}
  _i\langle \psi(c_u) | H_{E,i} | \psi(c_u)\rangle_i = 0
\end{equation}
and
\begin{equation}
_i\langle \psi(c_u) | H_{E,i}^2 | \psi(c_u)\rangle_i = |E| \le  \tau n
\end{equation}
When we measure this memory register, our probability of measuring the correct value is given by
  \begin{align}
    p_{\mathrm{succ}}^i &= |_i\langle \psi(c_u) | e^{-i S\theta H_{E,i}/2} | \psi(c_u)\rangle |^2 \\
    &\ge \cos^2 \left ( \frac{S\theta}{2} \sqrt{\tau n} \right )
  \end{align}
  where the inequality follows from Fleming's bound~\cite{fleming1973unitaritybound} so long as \(\frac{S\theta}{2} \sqrt{\tau n} \le \pi/2\). For our majority vote decoder to work, we need to obtain the correct result in a given register with probability \(p_{\mathrm{succ}}^i \ge 1/2 + \delta_{\tau}\) with \(\delta_{\tau} > 0\); this is true when \(\tau < 1/16\).
  Finally, given a fixed \(u\) and \(E\), our memory registers are in a product state at the end of the protocol, and the measurement results are independent. A standard Chernoff bound then yields that the probability of at most \(B/2\) of the registers yielding the right result is
  \begin{equation}
    \Pr[\hat{u} \ne u |u, E] \le e^{-2 \delta_{\tau}^2 B}.
  \end{equation}
  Thus, for \(B = \omega(1)\), \(p_{\mathrm{fail}} \to 0\) in the limit of large code size.
\end{proof}

With one more step, we may lift this to a threshold result when the transmitted codeword is subject to the binary symmetric channel.

\begin{theorem}
  \label{thm:threshold-qq}
  Consider a QQ receiver using the binary simplex code with even \(k\), \(B\cdot k\) memory qubits, \(B \cdot S\) sensor qubits, and signal strength \(\theta = \frac{2\pi}{S 2^{k/2}}\). If the bits of the transmitted codeword are flipped independently with probability \(p\), for fixed \(p < 1/16\), then
  \begin{equation}
    \lim_{k \to \infty} p_{\mathrm{fail}} = 0
  \end{equation}
  so long as \(B = B(k) = \omega(1)\), where \(p_{\mathrm{fail}}\) is the receiver failure probability.
\end{theorem}

\begin{proof}
  We use the same decoder as in the proof of Theorem~\ref{thm:adversarial-qq}. Fix \(p < \frac{1}{16}\), and pick \(\tau\) such that \(p < \tau < \frac{1}{16}\). Let \(W\) be a random variable denoting the weight of the transmission error. By Hoeffding's inequality, we get that
  \begin{equation}
    \Pr[W \ge \tau n] \le e^{-2 n (\tau-p)^2}.
  \end{equation}
  Combining this with Theorem~\ref{thm:adversarial-qq} gives that
  \begin{equation}
    p_{\mathrm{fail}} \le e^{-2 n (\tau-p)^2} + e^{-2 \delta_{\tau}^2 B}.
  \end{equation}
  which goes to zero so long as \(B = \omega(1)\) as desired.
\end{proof}

\begin{remark}
\label{rem:remark}
  By taking \(S=k\) and \(B = \lceil \log k \rceil\), we satisfy the requirements of Theorem~\ref{thm:threshold-qq} while using \(O(\log n \log \log n)\) total qubits. By Proposition~\ref{prop:qc-stochastic}, any QC receiver with the same number of qubits and same \(\theta\) will have success probability approaching zero.
\end{remark}

\subsection{Finite-Size Performance}
\label{sec:simulations}
\begin{figure}[ht]
  \centering
  \includegraphics[width=\columnwidth]{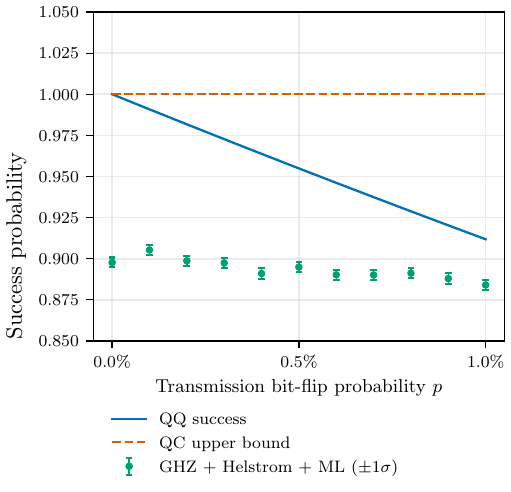}
  \caption{Finite-size decoding performance for the \([63,6,32]\) simplex code at total qubit budget \(Q=40\), with \(B=1\), \(S=Q-k\), and \(\theta=2\pi/(S2^{k/2})\). The solid blue curve is the exact QQ success probability, and the dashed orange curve is the QC upper bound obtained by numerically inverting Fano's inequality. Green points show the GHZ--Helstrom receiver followed by classical maximum-likelihood decoding, with one-standard-error bars from \(10^4\) trials per point. The QC bound is one throughout the plotted range.}
  \label{fig:q40-results}
\end{figure}
We now compare decoding success probabilities at finite code sizes and qubit budgets. For all comparisons, we use a single QQ register, \(B=1\), with \(k\) memory qubits and \(S=Q-k\) sensor qubits, and set
\begin{equation}
  \theta=\frac{2\pi}{(Q-k)2^{k/2}}, \qquad \Phi=Q\theta.
\end{equation}
The QC receivers use the same total number of qubits \(Q\) and signal strength \(\theta\), allocating all \(Q\) qubits to sensing. Figures~\ref{fig:q40-results} and~\ref{fig:q80-results} compare QQ performance with a QC upper bound and a concrete QC strategy for the \([63,6,32]\) simplex code at \(Q=40\) and the \([4095,12,2048]\) simplex code at \(Q=80\), respectively. In both cases, transmission bits flip independently with probabilities \(p\in[0,0.01]\). Additional code sizes and qubit budgets are considered in Appendix~\ref{sec:more-sims}.

For the QQ receiver, we evaluate the exact average success probability using Proposition~\ref{prop:qq-analytic-success-prob}. For the QC upper bound, we numerically invert the full Fano inequality used in the proof of Proposition~\ref{prop:qc-stochastic}. Specifically, with the per-slot information bound \(C=C_{QC}(\Phi,p)\), we numerically solve
\begin{equation}
  k-h_2(s)-(1-s)\log(2^k-1)=nC
\end{equation}
on \(s\in[2^{-k},1]\). The left-hand side increases from zero to \(k\) on this interval, giving an upper bound on success whenever \(0\le nC<k\); when \(nC\ge k\), we use the trivial bound \(s\le1\). For more details on the derivation of this equation, see Appendix~\ref{sec:proof-qc-stochastic}.

We also simulate a concrete non-adaptive QC receiver. At each time slot it prepares \((\ket{0^Q}+\ket{1^Q})/\sqrt{2}\) and performs the equal-prior Helstrom measurement to estimate the received bit. The two possible GHZ output states have overlap \(\lvert\cos(\Phi/2)\rvert\), giving a symmetric measurement error probability \(\epsilon=[1-\sin(\Phi/2)]/2\) in the regime \(\Phi\in[0,\pi]\) considered here. Composing transmission and measurement errors yields an effective binary symmetric channel with crossover probability
\begin{equation}
  p(1-\epsilon)+(1-p)\epsilon
  =\frac{1-(1-2p)\sin(\Phi/2)}{2}.
\end{equation}
In each trial, we sample transmission and measurement errors independently, then apply classical maximum-likelihood decoding to the measured bits \(y(x)\). For the simplex code, this decoder chooses the message maximizing
\begin{equation}
  \label{eq:simplex-ml-decoder}
  \sum_{x\ne0}(-1)^{y(x)\oplus u\cdot x}
  =n-2|y \oplus c_u|
\end{equation}
which is equivalent to finding the codeword with the minimum Hamming distance to the received bit string. Ties are broken uniformly at random. The linearity of the code, symmetry of the channel, and tie breaking strategy result in a success probability that is independent of the transmitted codeword, and thus for each trial we use the all zeros codeword, and compute Equation~\eqref{eq:simplex-ml-decoder} with a fast Walsh-Hadamard transform, setting the untransmitted \(x=0\) entry to zero.
Each plotted estimate uses \(10^4\) independent trials.

Figure~\ref{fig:q40-results} shows that, for \(Q=40\) and \(k=6\), the QQ receiver outperforms the simulated GHZ--Helstrom receiver with classical maximum-likelihood decoding at every sampled noise rate. The QC upper bound remains at one throughout the plotted range, so this comparison does not establish an advantage over every QC receiver. Improving the QC upper bound to obtain a provable advantage at these parameters remains an interesting avenue for future work.

Figure~\ref{fig:q80-results} establishes a separation at a finite qubit budget and codelength: for \(Q=80\) and \(k=12\), the QQ success probability exceeds the QC upper bound at every sampled noise rate. For example, at \(p=0.001\), the QQ receiver succeeds with probability approximately \(99.0\%\), while every QC receiver has success probability at most \(86.5\%\). The simulated GHZ strategy achieves \((53.0\pm0.5)\%\).

\begin{figure}[ht]
  \centering
  \includegraphics[width=\columnwidth]{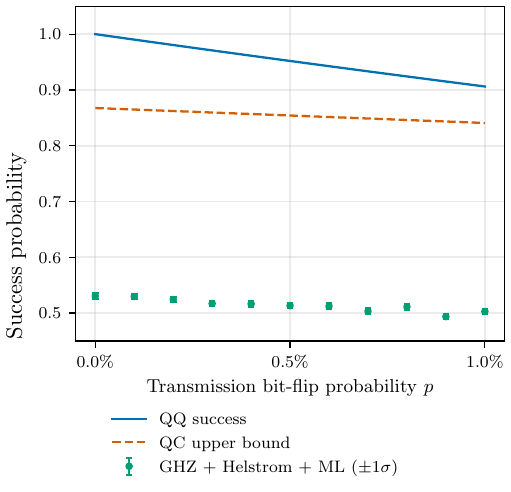}
  \caption{Finite-size decoding performance for the \([4095,12,2048]\) simplex code at \(Q=80\), with the same curve and marker conventions as Fig.~\ref{fig:q40-results}. The exact QQ success probability lies above the QC upper bound at every sampled transmission error probability.}
  \label{fig:q80-results}
\end{figure}

\section{Conclusion}
\label{sec:conclusion}
In this work, we have shown that a persistent quantum memory and quantum processing can improve the decoding of binary classical codes communicated via sequential weak \(Z\) rotations on qubits. For binary simplex codes, we construct receivers with a quantum memory that achieve reliable decoding in regimes where every receiver retaining only classical information between arrivals has vanishing success probability with the same qubit count and \(Z\) rotation angle. We additionally find that the advantage holds for finite simplex code sizes and receiver qubit counts.

Our constructions and bounds prioritize simplicity, leaving room for the improvement of constant factors and possibly finite size performance. A more interesting question is whether a similar advantage can be obtained for codes with higher rates. Simplex codes use \(2^k-1\) bits to encode only \(k\) message bits; extending our results to codes with higher rates would make quantum codeword sensing a more practical candidate for actual communication scenarios.

Another direction to pursue is robustness to noise in the receiver's qubits. Intriguingly, in the QQ protocol described in the work, the memory and sensor begin and end each time slot in a product state, in which the sensor qubits are all \(\ket{0}\). Thus, the sensor only needs to maintain coherence for the duration of a single bit arrival, raising the possibility of protecting the memory qubits in a quantum error correcting code while leaving the sensor qubits bare, although error propagation from the sensor to the memory would still need to be controlled. Finally, it would be interesting to investigate if the ideas of quantum codeword sensing could prove useful for temporal signal classification tasks beyond classical codeword recovery.

\section{Acknowledgments}
J.B. acknowledges support from the MIT Center for Quantum Engineering--Laboratory for Physical Sciences Doc Bedard Fellowship. Z.E.C acknowledges support from the NSF Graduate Research Fellowship Program and the Draper Scholarship. I.L.C. acknowledges support in part from the NSF Center for Ultracold Atoms.

\section{AI Use Statement}
The initial idea of using a temporal version of quantum trajectory sensing for classical error correction arose during a conversation between the authors in March of 2026. During the following summer, the authors tested ideas and produced initial constructions and bounds on classical memory receivers using various models of Claude Opus, Claude Fable, and GPT 5.6 Sol.  These constructions and bounds were then refined and simplified through further conversations with LLMs, including Claude Opus 5.5 and GPT 6 Astra. All proofs and constructions were verified, digested, and re-written by the authors, and the authors take full responsibility for the contents of this work.

\bibliographystyle{apsrev4-2}
\bibliography{bibliography}
\appendix
\onecolumngrid

\section{Deferred Proofs}
\label{sec:deferred-proofs}
\subsection{Proof of Lemma~\ref{lem:general-quadratic-decoder}}
\label{app:general-quadratic-decoder}

\begin{proof}
We follow the same phase-matching argument as in Section~\ref{sec:describe-protocol}. All vectors and matrix operations below are over \(\mathbb{F}_2\): addition is modulo two, so two identical terms cancel. Sums of the real signs \((-1)^{b(x)}\) use ordinary addition. First, we show that \(P^{-1}\), like \(P\), is symmetric with zero diagonal. Symmetry follows from
\[
(P^{-1})^{\mathsf T}=(P^{\mathsf T})^{-1}=P^{-1}.
\]
For any bitstring \(y\), the off-diagonal terms in \(y^{\mathsf T}Py\) occur in identical pairs and cancel modulo two, while the diagonal terms vanish because \(P_{ii}=0\). Thus, \(y^{\mathsf T}Py=0\) for every \(y\). Applying this observation to the bitstring \(P^{-1}y\), and using the symmetry of \(P^{-1}\), gives
\[
y^{\mathsf T}P^{-1}y
=(P^{-1}y)^{\mathsf T}P(P^{-1}y)=0.
\]
Choosing \(y\) to have a single 1 in position \(i\) shows that \((P^{-1})_{ii}=0\), as required.

The function \(b\) in Eq.~\eqref{eq:general-quadratic-offset} is quadratic because each term is a product of two input bits. To understand how it changes when we replace \(x\) by \(x\oplus y\), we expand these products. This separates terms involving only \(x\), terms involving only \(y\), and cross terms involving one bit from each string:
\begin{equation}
\label{eq:general-offset-polar}
\begin{aligned}
b(x\oplus y)
&=\bigoplus_{i<j}(P^{-1})_{ij}(x_i\oplus y_i)(x_j\oplus y_j)\\
&=\bigoplus_{i<j}(P^{-1})_{ij}
  (x_i x_j\oplus x_i y_j\oplus y_i x_j\oplus y_i y_j)\\
&=b(x)\oplus b(y)\oplus
  \bigoplus_{i<j}(P^{-1})_{ij}(x_i y_j\oplus y_i x_j)\\
&=b(x)\oplus b(y)\oplus x^{\mathsf T}P^{-1}y.
\end{aligned}
\end{equation}
The last equality uses the symmetry and zero diagonal of \(P^{-1}\): the two cross terms for each pair \(i<j\) are precisely the \(ij\) and \(ji\) terms of \(x^{\mathsf T}P^{-1}y\).

We next obtain the Fourier identity which will replace Eq.~\eqref{eq:quadratic-fourier-identity}. At a label \(z\), the normalized Walsh--Hadamard transform adds the signs \((-1)^{b(x)}\), reversing each sign for which the parity \(x\cdot z\) equals 1, and multiplies the sum by \(2^{-k/2}\). Our goal is to show that the result is always \(+1\) or \(-1\), and to determine how its sign depends on \(z\). All sums over bitstrings below are taken over \(\mathbb{F}_2^k\).

First take \(z=0\), so that every parity \(x\cdot z\) is zero. We need the magnitude of the real sum \(\sum_x(-1)^{b(x)}\). Squaring it introduces two inputs of \(b\); writing these inputs as \(x\) and \(x\oplus y\) will let us apply Eq.~\eqref{eq:general-offset-polar}. Explicitly,
\begin{equation}
\begin{aligned}
\left(\sum_x(-1)^{b(x)}\right)^2
&=\sum_x\sum_y(-1)^{b(x)\oplus b(y)}\\
&=\sum_y\sum_x(-1)^{b(x)\oplus b(x\oplus y)}\\
&=\sum_y\sum_x(-1)^{b(y)\oplus x^{\mathsf T}P^{-1}y}\\
&=\sum_y(-1)^{b(y)}\sum_x(-1)^{x^{\mathsf T}P^{-1}y}.
\end{aligned}
\end{equation}
In the second equality, we replace the second summation variable by \(x\oplus y\) and interchange the two finite sums. For each fixed \(x\), this change of variable is a bijection on the possible bitstrings, so the double sum is unchanged. The third equality uses Eq.~\eqref{eq:general-offset-polar}, with the two copies of \(b(x)\) canceling.

To evaluate the inner sum, first take \(y=0\). Then every exponent is zero, and the sum equals \(2^k\). If \(y\ne0\), invertibility implies \(P^{-1}y\ne0\), so there is some coordinate \(i\) with \((P^{-1}y)_i=1\). Flipping \(x_i\) changes \(x^{\mathsf T}P^{-1}y\) by 1, pairing each summand with one of the opposite sign. These pairs cancel, giving
\[
\sum_x(-1)^{x^{\mathsf T}P^{-1}y}
=\begin{cases}
2^k,&y=0,\\
0,&y\ne0.
\end{cases}
\]
Only \(y=0\) therefore contributes to the outer sum, and its prefactor is \((-1)^{b(0)}=1\). Consequently,
\[
\left(\sum_x(-1)^{b(x)}\right)^2=2^k,
\qquad
2^{-k/2}\sum_x(-1)^{b(x)}=\pm1.
\]

We can now evaluate the Fourier transform at an arbitrary memory label \(z\). We want to absorb the parity term \(x\cdot z\) into a shifted input of \(b\). Equation~\eqref{eq:general-offset-polar} produces the cross term \(x^{\mathsf T}P^{-1}y\), so choosing \(y=Pz\) makes this term exactly \(x\cdot z\). With this choice, the identity becomes
\[
b(x\oplus Pz)=b(x)\oplus b(Pz)\oplus x\cdot z.
\]
Rearranging this identity gives \(b(x)\oplus x\cdot z=b(x\oplus Pz)\oplus b(Pz)\). Thus,
\begin{equation}
\begin{aligned}
2^{-k/2}\sum_x(-1)^{b(x)\oplus x\cdot z}
&=2^{-k/2}\sum_x(-1)^{b(x\oplus Pz)\oplus b(Pz)}\\
&=(-1)^{b(Pz)}2^{-k/2}\sum_x(-1)^{b(x\oplus Pz)}\\
&=(-1)^{b(Pz)}2^{-k/2}\sum_x(-1)^{b(x)}.
\end{aligned}
\end{equation}
In the second equality, \(b(Pz)\) is independent of \(x\), so its sign factors out of the sum. In the last equality, shifting \(x\) by \(Pz\) simply permutes the bitstrings being summed over. The remaining normalized sum is the same \(\pm1\) found above, independent of \(z\). Thus, the normalized Fourier transform is either \(+1\) or \(-1\) at every \(z\). We define \(q_P(z)\) to be 0 in the first case and 1 in the second. At \(z=0\), this means \((-1)^{q_P(0)}=2^{-k/2}\sum_x(-1)^{b(x)}\). Our calculation shows that the sign at \(z\) is the sign at zero multiplied by \((-1)^{b(Pz)}\), so
\[
q_P(z)=b(Pz)\oplus q_P(0).
\]
Thus, \(q_P(z)\) differs from \(b(Pz)\) only by a constant bit, independent of \(z\). In this notation, the Fourier identity becomes
\begin{equation}
\label{eq:general-quadratic-fourier}
(-1)^{q_P(z)}=2^{-k/2}\sum_x(-1)^{b(x)\oplus x\cdot z}.
\end{equation}

As in the worked example, the left side is the phase factor we want \(D_b\) to impart to \(\ket{z}\), while the right side contains the activation function \(b(x)\oplus x\cdot z\). Apply \((-1)^t=1-2t\), valid for any bit \(t\), to both sides:
\begin{equation}
\begin{aligned}
1-2q_P(z)
&=2^{-k/2}\sum_x\left[1-2\bigl(b(x)\oplus x\cdot z\bigr)\right]\\
&=2^{k/2}-2^{1-k/2}\sum_{x\ne0}\bigl(b(x)\oplus x\cdot z\bigr).
\end{aligned}
\end{equation}
Here the \(x=0\) term in the bit sum vanishes because \(b(0)=0\). Rearranging gives
\begin{equation}
q_P(z)=\frac{1-2^{k/2}}{2}
+2^{-k/2}\sum_{x\ne0}\bigl(b(x)\oplus x\cdot z\bigr).
\end{equation}
Thus, writing \((-1)^{q_P(z)}=e^{i\pi q_P(z)}\), the desired phase factor becomes
\begin{equation}
\begin{aligned}
(-1)^{q_P(z)}
={}&e^{i\pi(1-2^{k/2})/2}\exp\!\left[\frac{i\pi}{2^{k/2}}\sum_{x\ne0}\bigl(b(x)\oplus x\cdot z\bigr)\right].
\end{aligned}
\end{equation}
The first factor is independent of \(z\). At \(\theta=2\pi/(S2^{k/2})\), the second is precisely the phase factor defining \(D_b\) in Eq.~(\ref{eq:offset-diagonal-action}). Thus, up to a global phase, \(D_b\ket{z}=(-1)^{q_P(z)}\ket{z}\).

It remains to show that this choice gives the desired message pattern. By Eq.~(\ref{eq:general-offset-polar}) and the symmetry of \(P\),
\begin{equation}
\begin{aligned}
q_P(z)\oplus q_P(z\oplus u)
&=b(Pu)\oplus(Pz)^{\mathsf T}P^{-1}(Pu)\\
&=b(Pu)\oplus z\cdot Pu.
\end{aligned}
\end{equation}
Consequently, the formal decomposition gives
\begin{equation}
\begin{aligned}
&D_bX[u]D_b^\dagger\ket{+}^{\otimes k}=(-1)^{b(Pu)}2^{-k/2}\sum_z(-1)^{z\cdot Pu}\ket{z}\\
&\qquad=(-1)^{b(Pu)}\ket{Pu}_X.
\end{aligned}
\end{equation}
The prefactor is a global phase. Since \(P\) is invertible, the measured string \(Pu\) uniquely determines \(u\).
\end{proof}
\subsection{Proof of Proposition \ref{prop:qc-stochastic}}
\label{sec:proof-qc-stochastic}
\begin{proof}
    We may consider the transmission of a single bit of the codeword as a single use of a classical-quantum channel. Let \(C_j\) be the random variable corresponding to the bit of the codeword and \(Y_j\) be the random variable corresponding to the measurement outcome in the \(j\)th time slot. Suppose \(\Pr[C_j = 0] = q\) and \(\Pr[C_j = 1] = (1-q)\). Furthermore, suppose the receiver initially prepares a \(Q\) qubit state \(\ket{\psi_0}\) (without loss of generality we assume they prepare a pure state, as using a mixed state is equivalent to a convex combination of pure state strategies, and thus has success probability upper bounded by the best pure state in the mixture). Let \(\ket{\psi_1} = R \ket{\psi_0}\). Accounting for the effects of errors, at the end of the time slot, if \(C_j = 0\), the receiver has the state
  \begin{equation}
    \sigma_0 = (1-p) \ket{\psi_0}\bra{\psi_0} + p \ket{\psi_1}\bra{\psi_1}
  \end{equation}
  and if \(C_j = 1\), the receiver has the state
  \begin{equation}
    \sigma_1 = p \ket{\psi_0}\bra{\psi_0} + (1-p) \ket{\psi_1}\bra{\psi_1}.
  \end{equation}
  We are interested in the maximum mutual information \(I(C_j; Y_j)\) over all possible \(q\) and POVMs, equivalently, \(\max_q I_{\mathrm{acc}}(\mathcal{E})\) for the ensemble \(\mathcal{E} = \{(q, \sigma_0), (1-q, \sigma_1)\}\). A standard upper bound would be the Holevo information \(\chi(\mathcal{E})\), but it turns out in this scenario we can do better. To this end, note that \(\sigma_0\) and \(\sigma_1\) are both confined to a two-dimensional subspace, and so
  \begin{align}
    F(\sigma_0, \sigma_1)^2 &= \mathrm{Tr}(\sigma_0 \sigma_1) + 2 \sqrt{\mathrm{det}(\sigma_0) \mathrm{det}(\sigma_1)} \\
    &= 1 - (1 - 2p)^2 (1 - \gamma^2)
  \end{align}
  where \(\gamma = |\langle \psi_0 | \psi_1 \rangle| \ge \cos (\Phi / 2)\) by Lemma~\ref{lem:overlap}. By Uhlmann's theorem there exist purifications \(\ket{\phi^0}\) and \(\ket{\phi^1}\) such that \(F(\sigma_0, \sigma_1) = \langle \phi^0 | \phi^1 \rangle\). Note that \(I_{\mathrm{acc}}(\mathcal{E}) \le I_{\mathrm{acc}}(\mathcal{E}')\) where \(\mathcal{E}' = \{(q, \ket{\phi^0}), (1-q, \ket{\phi^1})\}\), as the receiver can always replicate a measurement of a mixed state, using the purified state, by throwing away the purifying system. By \cite{ban1997accessibleinformation,sasaki1998quantumchannels}, the accessible information of an ensemble of two pure states with overlap \(\kappa\), maximized over input priors, equals \(1-h_2\big((1-\sqrt{1-\kappa^2})/2\big)\).  Applying this to \(\mathcal{E}'\), whose two states have overlap \(\kappa = F(\sigma_0,\sigma_1)\), and using \(\sqrt{1-F(\sigma_0,\sigma_1)^2} = (1-2p)\sqrt{1-\gamma^2}\) from above,
  \begin{align}
    I_{\mathrm{acc}}(\mathcal{E}) &\le 1 - h_2\left( \frac{1 - (1-2p)\sqrt{1-\gamma^2}}{2} \right) \\
                                  &\le 1 - h_2\left( \frac{1-(1-2p)\sin(\Phi/2)}{2} \right) \\
                                  &= C_{QC}(\Phi, p).
  \end{align}

  The scenario of a QC receiver, in which the receiver is required to prepare a new state at the start of each time slot, and measure at the end of each time slot, is equivalent to a classical-quantum channel in which only (potentially adaptive) product measurements are made. By the chain rule of mutual information,
  \begin{align}
    I(C^n;Y^n) &=\sum_{j=1}^n I(C^n;Y_j\mid Y_{<j})\\
&=\sum_{j=1}^n I(C_j;Y_j\mid Y_{<j})\\
&\le nC_{QC}(\Phi,p).
  \end{align}
  where we used that conditioned on the previous set of measurement results, \(I(C_j; Y_j | Y_{<j} = h) \le I_{\mathrm{acc}}(\mathcal{E}_{j,h}) \le C_{QC}(\Phi,p)\) (where \(\mathcal{E}_{j,h}\) is the ensemble obtained on the \(j\)th slot, given previous measurement results \(h\)).

  Let \(s = \Pr[u = \hat{u}]\). Using Fano's inequality, we obtain
  \begin{equation}
     n C_{QC}(\Phi, p) \ge I(C^n ; Y^n) \ge H(C^n) - h_2(s) - (1-s) \log (2^k-1)
\label{eq:single-q-mut-info-lower}
  \end{equation}
  Noting that \(h_2(s) \le 1\) yields that
  \begin{equation}
    \label{eq:simplified-fanos}
    s \le \frac{1 + n C_{QC}(\Phi, p)}{k} 
  \end{equation}
  Alternatively, we may use Pinsker's inequality to obtain
  \begin{equation}
    n C_{QC}(\Phi, p) \ge \frac{2}{\ln 2} \left ( s - \frac{1}{2^{k}}\right)^2
  \end{equation}
  and solving for \(s\) gives
  \begin{equation}
    \label{eq:pinsker-bound}
    s \le \frac{1}{2^k} + \sqrt{\frac{\ln 2}{2} n C_{QC}(\Phi, p)}
\end{equation}

Finally, we may slightly modify our analysis above for the case of a CC receiver. Here, \(\ket{\psi_0} = \bigotimes_{i=1}^Q \ket{\phi_i}\), and thus \( \gamma = |\langle \psi_0 | R | \psi_0\rangle| \ge \cos^Q(\theta/2)\). We may make the same argument about accessible information, for this ensemble, instead obtaining
  \begin{align}
    I_{\mathrm{acc}}(\mathcal{E}) &\le 1 - h_2\left ( \frac{1 - (1-2p)\sqrt{1 - \gamma^2}}{2} \right ) \\
                        & \le 1 - h_2 \left ( \frac{1-(1-2p)\sqrt{1 - \cos^{2Q}(\theta/2)}}{2} \right )  \\
    &= C_{CC}(\theta, p, Q).
  \end{align}
The rest of the argument follows, noting that now a single measurement result \(Y_j\) refers to the result of a POVM of the entire register.
\end{proof}

\subsection{QQ Receiver Success Probability}
\label{sec:qq-success-prob}

\begin{proposition}
  \label{prop:qq-analytic-success-prob}
  Consider the QQ receiver described in Section~\ref{sec:describe-protocol} for the simplex code with \(B = 1\), \(k\) even, and \(\alpha = S\theta/2 = \pi/2^{k/2}\). Suppose the transmitted codeword is subject to the binary symmetric channel with bit flip probability \(p\). Then, the probability that the decoder returns the correct message is
  \begin{equation}
    \label{eq:qq-b1-success}
    p_{\mathrm{succ}}= 2^{-k} + (1-2^{-k})|w|^{2^{k-1}} \left (\frac{1}{2} + \frac{1}{2} \cos \left (2^{k/2} \arg w  \right ) \right )
\qquad
w=1-p+pe^{2i\alpha}.
\end{equation}
\end{proposition}

Before the proof, a short technical lemma.
\begin{lemma}
  \label{lem:technical}
  Let \(g: \mathbb{F}_2^k \to \{-1, 1\}\) be bent, that is, its unnormalized Walsh-Hadamard transform: \(\hat{g}(z) = \sum_{x \in \mathbb{F}_2^k} g(x)(-1)^{z\cdot x}\) has constant magnitude, and let \(d \in \mathbb{F}_2^k\) with \(d \ne 0\). Then, for half of the values \(z \in \mathbb{F}_2^k\), \(\hat{g}(z) = \hat{g}(z\oplus d)\).
\end{lemma}
\begin{proof}
  Consider that
  \begin{align}
    \sum_z \hat{g}(z) \hat{g}(z \oplus d) &= \sum_x \sum_y \sum_z g(x) g(y) (-1)^{z\cdot x} (-1)^{(z\oplus d)\cdot y} \\
             &= \sum_{x,y} g(x) g(y) (-1)^{d\cdot y} \sum_z (-1)^{z\cdot(x\oplus y)} \\
                                  &= 2^k \sum_x g(x)^2 (-1)^{d\cdot x} \\
    &= 0
\end{align}
Since \(g\) is bent, \(\hat{g}(z)\hat{g}(z\oplus d) = 2^k\) when \(\hat{g}(z)=\hat{g}(z\oplus d)\) and \(\hat{g}(z)\hat{g}(z\oplus d) = -2^k\) otherwise. Thus, the number of \(z\) values for which they agree is equal to the number for which they don't.
\end{proof}

We now give a proof of Proposition~\ref{prop:qq-analytic-success-prob}.
\begin{proof}
Suppose message \(u\) is sent, and error \(E\) occurs.
As in the proof of Theorem~\ref{thm:adversarial-qq}, let
\begin{equation}
 H_{E} = \sum_{x \in E} g(x) (-1)^{c_u(x)} Z[x]
\end{equation}
where we have defined \(g(x) = (-1)^{b(x)}\).
The probability of successful decoding is given by
\begin{align}
 | \langle \psi(c_u)| e^{i\alpha H_E} |\psi(c_u) \rangle |^2 &= \frac{1}{2^{2k}} \left | \sum_{z} \exp\left [i \alpha \sum_{x \in E} g(x) (-1)^{c_u(x)} (-1)^{x \cdot z}\right ] \right|^2 \\
        &= \frac{1}{2^{2k}} \left | \sum_{z} \exp\left [i \alpha \sum_{x \in E} g(x)(-1)^{x \cdot z}\right] \right|^2
\end{align}
where in the last line we used that \(c_u(x) = u\cdot x\) and reindexed the \(z\) variable. Note that this is independent of the message \(u\). Thus, the total success probability is given by
\begin{align}
       p_{\mathrm{succ}} &= 2^{-2k}  \sum_{z,d} \mathbb{E}_E \exp \left [ i\alpha \sum_{x \in E} g(x) (-1)^{z\cdot x} (1 - (-1)^{d\cdot x}) \right ] \\
                         &= 2^{-k} + 2^{-2k} \sum_{z} \sum_{d\ne 0} \sum_{E \subseteq [n]} p(E) \prod_{x\in E} \exp \left [i \alpha g(x) (-1)^{z\cdot x}(1-(-1)^{d\cdot x}) \right ] \\
       &= 2^{-k} + 2^{-2k} \sum_{z} \sum_{d \ne 0} \prod_{x \in \mathbb{F}_2^k} \left ( (1-p) + p \exp \left [ i \alpha g(x) (-1)^{z\cdot x}(1 - (-1)^{d\cdot x}) \right ] \right ).
\end{align}
Consider that
\begin{equation}
       (1-p) + p \exp \left [ i \alpha g(x) (-1)^{z\cdot x}(1 - (-1)^{d\cdot x}) \right ] = \begin{cases}
         1 &\text{ if } d\cdot x = 0 \\
         1 - p + pe^{\pm 2i\alpha} &\text{ otherwise}
\end{cases}
\end{equation}
Let \(w = 1 - p + pe^{2i\alpha}\) and for fixed \(z\) and \(d\), with \(d\neq 0\), let \(a_x = g(x)(-1)^{z\cdot x}[1-(-1)^{d\cdot x}]\). Then, 
\begin{equation}
  \prod_{x} [1-p + pe^{i\alpha a_x}] = |w|^{2^{k-1}} \exp \left [ i \frac{\arg w}{2} \sum_x a_x \right ],
\end{equation}
where we used that there are \(2^{k-1}\) values of \(x\) for which \(d\cdot x = 1\), and
\begin{align}
 \sum_x a_x &= \sum_x g(x)(-1)^{z\cdot x}[1-(-1)^{d\cdot x}] \\
  &= \hat{g}(z) - \hat{g}(z \oplus d).
\end{align}
By Lemma~\ref{lem:technical}, on half of \(z \in \mathbb{F}_2^k\), this sum is zero, and on the other half, it is \(\pm 2 \cdot 2^{k/2}\). Note that the substitution \(z \mapsto z \oplus d\) changes the sign, so the positive and negative cases each occur for half of the possible values of \(z\) giving a nonzero sum. Thus, the whole average becomes
\begin{align}
  p_{\mathrm{succ}} &= 2^{-k} + (1-2^{-k})|w|^{2^{k-1}} \left (\frac{1}{2} + \frac{1}{4} e^{-i 2^{k/2} \arg w}  + \frac{1}{4} e^{i 2^{k/2} \arg w}\right ) \\
  &= 2^{-k} + (1-2^{-k})|w|^{2^{k-1}} \left (\frac{1}{2} + \frac{1}{2} \cos \left ( 2^{k/2} \arg w  \right ) \right )
\end{align}
\end{proof}

\section{Additional Finite-Size Performance Results}
\label{sec:more-sims}
In this section we provide comparisons between QQ and QC receivers using the parameters and methods described in Section~\ref{sec:simulations}. Figures~\ref{fig:q40-full-results}--\ref{fig:q80-full-results} show the full comparisons for \(Q\in\{40,60,80\}\) and \(k\in\{6,8,10,12,14,16\}\). Each panel is labeled by the simplex code parameters \([n,k,d]=[2^k-1,k,2^{k-1}]\).

\begin{figure}[htbp]
  \centering
\includegraphics[width=0.8\textwidth]{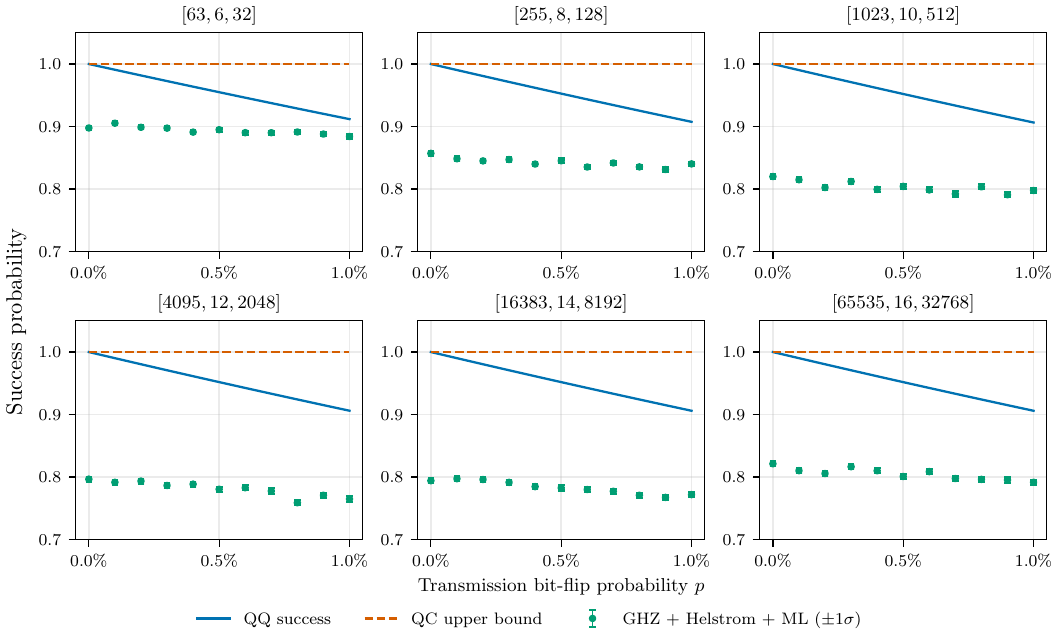}
  \caption{Finite-size decoding performance at total qubit budget \(Q=40\). Each panel shows a different simplex code dimension \(k\), with \(B=1\), \(S=Q-k\), and \(\theta=2\pi/(S2^{k/2})\). The solid blue curve is the exact QQ success probability, and the dashed orange curve is the QC upper bound obtained by numerically inverting Fano's inequality. Green points show the GHZ--Helstrom receiver followed by classical maximum-likelihood decoding, with one-standard-error bars from \(10^4\) trials per point. For all data points, the error bars are smaller than the data marker. The QC bound is one throughout the plotted range.}
  \label{fig:q40-full-results}
\end{figure}

\begin{figure}[htbp]
  \centering
  \includegraphics[width=0.8\textwidth]{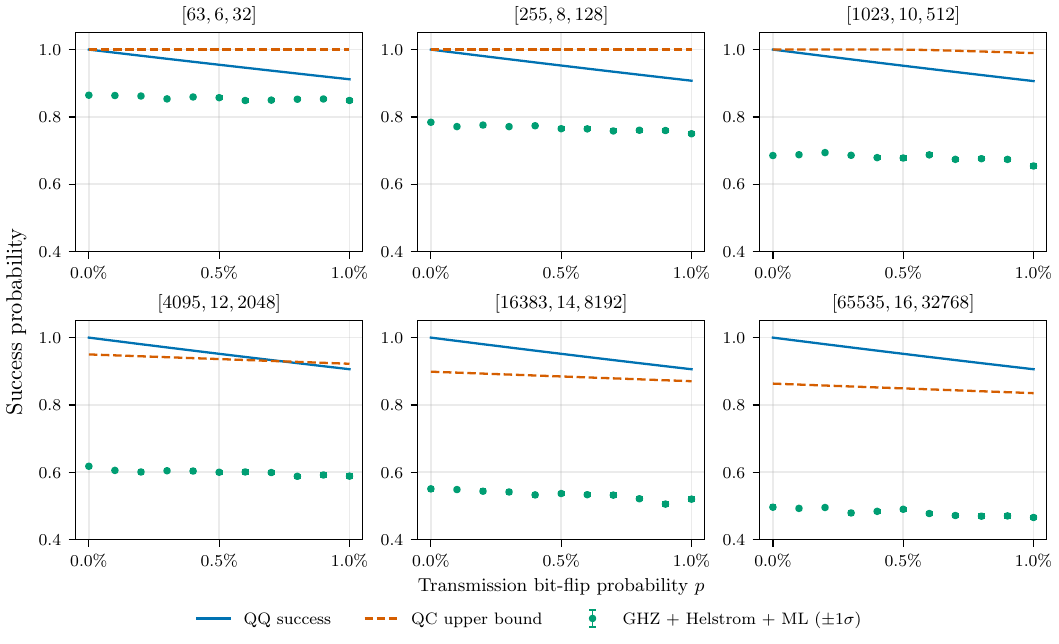}
  \caption{Finite-size decoding performance at total qubit budget \(Q=60\). Each panel shows a different simplex code, with the same receiver parameters, curve conventions, and Monte Carlo error bars as Fig.~\ref{fig:q40-results}.}
  \label{fig:q60-full-results}
\end{figure}

\begin{figure}[htbp]
  \centering
\includegraphics[width=0.8\textwidth]{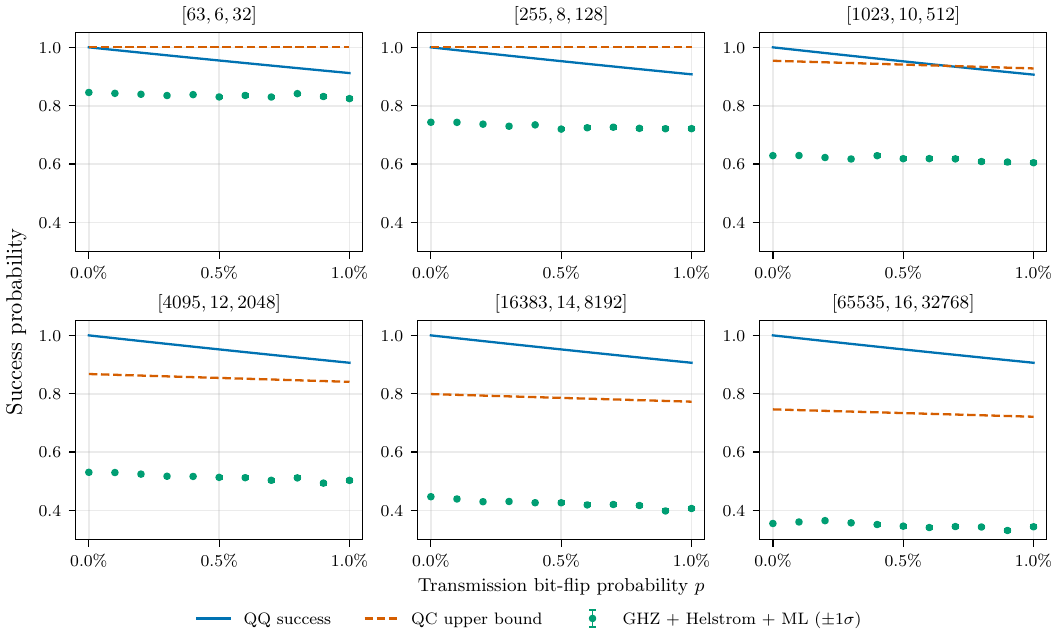}
  \caption{Finite-size decoding performance at \(Q=80\), with the same curve and marker conventions as Fig.~\ref{fig:q40-full-results}. For \(k=12,14,16\), the exact QQ success probability lies above the QC upper bound at every sampled transmission error probability. The \(k=10\) panel also shows a separation at smaller transmission error probabilities.}
  \label{fig:q80-full-results}
\end{figure}
\end{document}